\documentclass{amsart}

\usepackage{amsfonts}
\usepackage{amsfonts}
\usepackage{amssymb}
\usepackage{amsthm}
\usepackage{amssymb}
\usepackage{amsmath}
\usepackage{epstopdf}
\usepackage[utf8]{inputenc}
\usepackage{graphicx}
\usepackage[english]{babel}
\usepackage{latexsym}
\usepackage{epsfig}
\usepackage{epstopdf} 
\usepackage{xspace}
\usepackage{xcolor}
\usepackage[mathscr]{eucal}
\usepackage{color}
\usepackage{float}
\usepackage{enumerate}

\newtheorem{theorem}{Theorem}
\newtheorem{lemma}[theorem]{Lemma}

\newtheorem{definition}[theorem]{Definition}
\newtheorem{example}[theorem]{Example}

\newtheorem{corollary}[theorem]{Corollary}
\newtheorem{proposition}[theorem]{Proposition}
\newtheorem{assumption}[theorem]{Assumption}

\newtheorem{remark}[theorem]{Remark}

\newcommand{\hmin}{h_{\min}}
\newcommand{\hmax}{h_{\max}}

\newcommand{\CH}[1]{{\color{black} #1}}
\newcommand{\Note}[1]{{\color{black} #1}}

\newcommand{\expect}[1]{\mathbb{E}\left[  {#1} \right]}

\newcommand{\Dt}{\Delta t}

\newcommand\bc{\begin{center}}
\newcommand\ec{\end{center}}
\newcommand\bi{\begin{itemize}}
\newcommand\ei{\end{itemize}}
\newcommand\be{\begin{enumerate}}
\newcommand\ee{\end{enumerate}}
\newcommand\bd{\begin{definition}}
\newcommand\ed{\end{definition}}
\newcommand\bt{\begin{theorem}}
\newcommand\et{\end{theorem}}
\newcommand\bp{\begin{proposition}}
\newcommand\ep{\end{proposition}}
\newcommand\bcor{\begin{corollary}}
\newcommand\ecor{\end{corollary}}
\newcommand\bx{\begin{exercise}}
\newcommand\ex{\end{exercise}}
\newcommand\beg{\begin{example}}
\newcommand\eeg{\end{example}}
\newcommand\bl{\begin{lemma}}
\newcommand\el{\end{lemma}}
\newcommand\bea{\begin{eqnarray*}}
\newcommand\eea{\end{eqnarray*}}

\newcommand{\Cp}[1]{\mathbf{\upsilon}_{#1}}

\begin{document}



\title[An adaptive method for CIR]{The role of adaptivity in a numerical method for the Cox-Ingersoll-Ross model}

\author[C. Kelly]{C\'onall Kelly$^\ast$}
\address{School of Mathematical Sciences\\ University College Cork\\ Republic of Ireland}
\address{Email: conall.kelly@ucc.ie}

\author[G. J. Lord]{Gabriel J. Lord}
\address{Department of Mathematics\\
Radboud University\\
Netherlands}
\address{Email: gabriel.lord@ru.nl}

\author[H. Maulana]{Heru Maulana}
\address{School of Mathematical Sciences\\ University College Cork\\ Republic of Ireland}
\address{Postal address: Mathematics Department\\ Faculty of Mathematics and Natural Sciences\\ Universitas Negeri Padang\\ Republic of Indonesia}
\address{Email: herumaulana@fmipa.unp.ac.id}

\thanks{$^\ast$Corresponding author.}


\date{\today}

\maketitle
\begin{abstract}
We demonstrate the effectiveness of an adaptive explicit Euler method for the approximate solution of the Cox-Ingersoll-Ross model. This relies on a class of path-bounded timestepping strategies which work by reducing the stepsize as solutions approach a neighbourhood of zero. The method is hybrid in the sense that a convergent backstop method is invoked if the timestep becomes too small, or to prevent solutions from overshooting zero and becoming negative. Under parameter constraints that imply Feller's condition, we prove that such a scheme is strongly convergent, of order at least $1/2$. Control of the strong error is important for multi-level Monte Carlo techniques. Under Feller's condition we also prove that the probability of ever needing the backstop method to prevent a negative value can be made arbitrarily small. Numerically, we compare this adaptive method to fixed step implicit and explicit schemes, and a novel semi-implicit adaptive variant. We observe that the adaptive approach leads to methods that are competitive in a domain that extends beyond Feller's condition, indicating suitability for the modelling of stochastic volatility in Heston-type asset models. 
\end{abstract}



\section{Introduction}

The Cox-Ingersoll-Ross (CIR) process, used for example in the pricing of interest rate derivatives and as a model of stochastic volatility, is given by the following It\^o-type stochastic differential equation (SDE), 
\begin{equation}\label{eq:equ1}
dX(t)= \kappa \left( \lambda - X(t) \right) dt+ \sigma \sqrt{X(t)} dW(t), \ t \in [0,T];\ X(0)=X_0>0,
\end{equation}
where $W(t)$ is a Wiener Process, $\kappa$, $\lambda$, and $\sigma$ are positive parameters, and $T\in[0,\bar T]$ for some fixed $\bar T<\infty$. Solutions of \eqref{eq:equ1} are almost surely (a.s.) non-negative: in general they can achieve a value of zero but will be reflected back into the positive half of the real line immediately. Moreover, if  $2\kappa\lambda\geq \sigma^2$, referred to as the Feller condition, solutions will be a.s. positive. No closed form solution of \eqref{eq:equ1} is available, though  $X(t)$ has (conditional upon $X(s)$ for $0\leq s<t$) a non-central chi-square distribution with $\lim_{t\to\infty}\mathbb{E}[X(t)]=\lambda$ and  $\lim_{t\to\infty}$Var$[X(t)]=\lambda\sigma^2/2\kappa$; see \cite{5}.  

For Monte Carlo estimation, exact sampling from the known conditional distribution of $X(t)$ is possible but computationally inefficient and potentially restrictive if the Wiener process of \eqref{eq:equ1} is correlated with that of another process: see \cite{Alphonsi2005,3,8,15}. Consequently a substantial literature has developed on the efficient numerical approximation of solutions of \eqref{eq:equ1}; we now highlight the parts which are relevant to our analysis.

An approach that seeks to directly discretise \eqref{eq:equ1} using some variant of the explicit Euler-Maruyama method leads to schemes of the form
\begin{equation}\label{eq:expEul}
\begin{split}
\widetilde{V}_{n+1} &= g_0\left(\widetilde{V}_n\right) + \Dt\kappa\left(\lambda- g_1\left(\widetilde{V}_n\right)\right) + \sigma \sqrt{g_2\left(\widetilde{V}_n\right)} \Delta W_n;\\
V_{n+1}&=g_3\left(\widetilde{V}_{n+1}\right);\quad \widetilde{V}_0=V_0,
\end{split}
\end{equation}
for given functions $g_0$, $g_1$,$g_2$ and $g_3$. These functions are selected to ensure that the diffusion coefficient remains real-valued (so that \eqref{eq:expEul} is well defined), and to preserve the non-negativity of solutions. This approach seeks to accommodate the non-Lipschitz (square-root) diffusion, which facilitates overshoot when the solutions are close to zero, but it introduces additional bias to the approximation. A survey of choices common in practice may be found in \cite{15}, and we present a similar selection in Table \ref{tab:my_label} using the convention $x^+:=\max\{0,x\}$. We highlight in particular the fully truncated method proposed in Lord et al~\cite{15}. While it was shown in that article that the method is strongly convergent in $L_1$, the rate of strong convergence has only been recently proved by Cozma \& Reisinger~\cite{Cozma}, who demonstrated a strong order of convergence $1/2$ in $L^{p}$, in the case where $2 \kappa \lambda >3 \sigma^2$, for $2\leq p<2\kappa\lambda/\sigma^2-1$. This method preserves the positivity of the underlying solutions of \eqref{eq:equ1}, and the authors of \cite{Cozma} state that it is arguably the most widely used in practice.

\begin{table}[]
    \centering
    \begin{tabular}{c|c|c|c|c}
         Method & $g_0(x)$ & $g_1(x)$ & $g_2(x)$ & $g_3(x)$\\ \hline
         Explicit Euler & $x$ & $x$ & $x$ & $x$\\ 
         Partially Truncated~\cite{6} & $x$ & $x$ & $x^+$ & $x$\\
         Fully Truncated~\cite{15} & $x$ &  $x^+$ & $x^+$ & $x^+$ \\
         Higham \& Mao~\cite{9} & $x$ & $x$ & $|x|$ & $x$ \\
    \end{tabular}
    \caption{Explicit Euler-Maruyama  variants.}
    \label{tab:my_label}
\end{table}

An alternative approach is to transform \eqref{eq:equ1} before discretisation to make the diffusion coefficient globally Lipschitz. For example, applying the Lamperti transform $Y=\sqrt{X}$ yields an auxiliary SDE in $Y$ with a state independent and therefore globally Lipschitz diffusion, but a drift coefficient that is unbounded when solutions are in a neighbourhood of zero.  This approach is effective: a fully implicit Euler discretisation over a uniform mesh that preserves positivity of solutions was proposed in \cite{Alphonsi2005} and shown to have uniformly bounded moments. A continuous time extension interpolating linearly between mesh points was shown to have strong $L_p$ order of convergence $1/2$ (up to a factor of $\sqrt{|\log(h)|}$) in \cite{7} when $2\kappa\lambda>p\sigma^2$, a continuous-time variant based on the same implicit discretisation was shown to have strong $L_p$ convergence of order $1$ when $4\kappa\lambda>3p\sigma^2$ in \cite{Alphonsi2013}, and in \cite{CJM2016} a variant which discretised the transformed SDE for $Y$ with an explicit projection method was shown to give strong $L_2$ convergence of order $1$ when $2\kappa\lambda>5\sigma^2$, of order $1/2$ when $3\sigma^2<2\kappa\lambda\leq 5\sigma^2$, and with an order on the interval $(1/6,(2\kappa\lambda-\sigma^2)/(4\kappa\lambda+2\sigma^2))$ when $2\sigma^2<2\kappa\lambda\leq 3\sigma^2$.

It is important to emphasise the distinction between weakly and strongly convergent numerical methods. Weakly convergence methods may be sufficient for the Monte Carlo estimation of some derivatives, and methods which converge weakly with high order for \eqref{eq:equ1} are known, see for example~\cite{Alfonsi2010}. However strongly convergent methods are required 
in order to take advantage of Multi-Level Monte Carlo variance reduction techniques; see \cite{Giles2015}.

In this article we show that a strongly convergent numerical scheme can be constructed by an application of the Lamperti transform to \eqref{eq:equ1} followed by an explicit Euler-Maruyama discretisation over a procedurally generated adaptive mesh. The purpose of the adaptivity is to manage the nonlinear drift response of the discrete transformed system (rather than local error control). A framework for this was introduced in \cite{12} for SDEs with one-sided Lipschitz drift and globally Lipschitz diffusion and extended to allow for monotone coefficients and a Milstein-type discretisation in \cite{13,KeLoSu2019} respectively. This framework imposes maximum and minimum timesteps $h_{\max}$ and $h_{\min}$ in a fixed ratio $\rho$ and requires the use of a backstop numerical method in the event that the timestepping strategy attempts to select a stepsize below $h_{\min}$. The introduction of \cite{12} provides a comprehensive review of the adaptive literature for SDEs.  

As in \cite{KeLoSu2019}, we will use here path-bounded strategies, this time designed to increase the density of timesteps when solutions approach zero, and we additionally require the backstop method to retake a step when the adaptive strategy overshoots the singularity at zero in the transformed equation. This latter is carried out without discarding samples from the Brownian path (preserving the trajectory), and without bridging (preserving efficiency).

We prove, when $\kappa\lambda>2\sigma^2$, that the order of strong convergence in $L_2$ is at least $1/2$. This parameter constraint implies the Feller condition and is technical, ensuring the finiteness of sufficiently many conditional inverse moments of solutions of \eqref{eq:equ1} (as described by \cite{2}). We separately prove that, under exactly the Feller condition, the probability of invoking the backstop method to avoid negative values can be made arbitrarily small by choosing $h_{\max}$ sufficiently small, and provide a practical method for doing so given a user defined tolerance level. The proof relies upon a finite partitioning of the sample space of trajectories induced by $h_{\max}$ and $h_{\min}$, which allows us to handle the randomness of the number of timesteps via the Law of Total Probability.

Numerically we compare the convergence and efficiency of our hybrid adaptive method with a semi-implicit adaptive variant, the fixed-step explicit method due to \cite{9}, and the transformed implicit fixed-step method proposed and analysed in \cite{Alphonsi2005,7,Alphonsi2013}, examining the parameter dependence of the numerical order of convergence in each case. The numerical convergence rates of adaptive methods are seen to outperform those of fixed-step methods over the entire domain where Feller's condition holds. Indeed we observe polynomial orders of convergence beyond this domain, indicating that these methods are also applicable to modelling stochastic volatility processes, for example in a Heston model.

Our results extend naturally to variants of \eqref{eq:equ1} with time dependent parameters (see Glasserman~\cite{8}) subject to the existence of inverse moments in that setting, and to shifted CIR models such as those found in~\cite{OrlandoEtAl2019} which allow for negative interest rates.

The structure of the article is as follows. In Section \ref{sec:prelim} we give the form of the SDE governing the Lamperti transform of \eqref{eq:equ1}, specify the constraints placed upon the parameters for the main strong convergence theorem, and examine the availability of conditional moment and inverse moment bounds under these constraints. In Section \ref{sec:adapt} we set up the framework for our random mesh, characterise the class of path-bounded timestepping strategies and define our adaptive numerical method. In Section \ref{sec:main} we present the two main theorems on strong convergence and positivity, providing illustrative examples in the latter case. Finally in Section \ref{sec:numerics} we numerically compare convergence and efficiency of several commonly used methods.

\section{Mathematical Preliminaries}\label{sec:prelim}
Throughout this article we let $(\mathcal{F}_t)_{t\geq 0}$ be the natural filtration of $W$. By using It\^o's formula and applying the transformation $Y= \sqrt{X}$ we get,
\begin{equation*}
d Y(t)= \left( \frac{4\kappa \left( \lambda - X_t \right) - \sigma^2}{8 \sqrt{X_t}} \right) dt+ \frac{\sigma}{2} dW_t,\ t\in [0,T]; \quad Y(0) = \sqrt{X_0}\in \mathbb{R}^+. \\ 
\end{equation*}
By then setting,
\begin{equation}
\nonumber
\alpha = \frac{4\kappa \lambda-\sigma^2}{8}, \ \beta=\frac{-\kappa}{2}, \ \gamma=\frac{\sigma}{2},
\end{equation}
we can write,
\begin{equation}\label{eq:equ3}
d Y(t)= \left( \frac{\alpha}{Y(t)}+\beta Y(t) \right) dt+ \gamma dW_t,\ t\in [0,T]; \quad Y(0) = \sqrt{X_0}\in \mathbb{R}^+, \\
\end{equation}
where $f(y)=\alpha/y+\beta y$ is not globally Lipschitz continuous, but when $\alpha>0$ it satisfies a one-sided Lipschitz condition with constant $\beta<0$:
\begin{equation*}
[f(x)-f(y)](x-y) \leq \beta (x-y)^2, \ \text{for all} \ x,y \in \mathbb{R}^+,
\end{equation*}
which can be seen by noting that 
\[
f(x)-f(y)=(x-y)\left[\beta-\frac{\alpha}{xy}\right].
\]
Meanwhile the diffusion coefficient $g(y)=\gamma$ is constant and therefore globally Lipschitz continuous. The SDE \eqref{eq:equ3} has integral form
\begin{equation}\label{CIRtInt}
Y(t)=Y(0)+\int_{0}^{t}\left( \frac{\alpha}{Y(s)}+\beta Y(s) \right)ds+\int_{0}^{t}\gamma dW_t,\quad t\geq 0.
\end{equation}

In order to ensure the a.s. positivity of solutions of \eqref{eq:equ3} and the boundedness of certain inverse moments of solutions of \eqref{eq:equ3}, we will also need the following assumption:

\begin{assumption}\label{assum:fel2}
Suppose that
\begin{equation}\label{eq:equ5str}
\kappa\lambda> 2\sigma^2.
\end{equation}
\end{assumption}

Eq. \eqref{eq:equ5str} implies the Feller Condition ($2\kappa\lambda\geq\sigma^2$; see, for example \cite[Chapter 9.9.2, p. 308]{17}),  which ensures that solutions of \eqref{eq:equ1}, and therefore \eqref{eq:equ3}, remain positive with probability one:
\begin{equation*}
\mathbb{P}[Y(t)>0, \ t\geq 0]=1.
\end{equation*}

Assumption \ref{assum:fel2} provides inverse moment bounds as follows:
\begin{lemma}\label{lem:fin}
Let $\left( Y(t) \right)_{t \in [0,T]}$ be a solution of \eqref{eq:equ3}, where Assumption \ref{assum:fel2} holds, and let $0 \leq t< s \leq T$. For any $Y(0)>0$, and 
for $1\leq p\leq 6$, there exists $C(p,T)>0$ such that
\begin{equation} \label{eq:equM-2}
\mathbb{E} \left[ \frac{1}{Y(s)^{p}} \biggl| \mathcal{F}_t \right] \leq \frac{C(p,T)}{Y(t)^p}, \quad a.s.
\end{equation}
\end{lemma}
\begin{proof}

Let $\left( X(t) \right)_{t \in [0,T]}$ be a solution of \eqref{eq:equ1} where Assumption \ref{assum:fel2} holds. By Lemma A.1 in Bossy \& Diop~\cite{2}, 
\begin{equation}\label{eq:fromBD}
\mathbb{E} \left[ \frac{1}{X(t)} \right] \leq \frac{e^{\kappa t}}{X_0}\quad \text{and}\quad \mathbb{E} \left[ \frac{1}{X(t)^p} \right] \leq \frac{C(p,T)}{X_0^p},
\end{equation} 
for some $C(p,T)$ and any $p$ such that $1<p<\frac{2\kappa\lambda}{\sigma^2}-1$. Assumption \ref{assum:fel2} ensures that $\frac{2\kappa\lambda}{\sigma^2}-1>3$, and since $Y(t)=\sqrt{X(t)}$,  \eqref{eq:equM-2} follows by Lemma A.1 in \cite{2} as it applies to conditional expectations, the former requiring an additional application of Jensen's inequality to the first inequality in \eqref{eq:fromBD}. 
\end{proof}

We also need the following bounds on positive moments of solutions of \eqref{eq:equ3}, which apply under Feller's condition and in particular under Assumption \ref{assum:fel2}.

\begin{lemma}\label{lem:posMB}
Let $\left( Y(t) \right)_{t \in [0,T]}$ be a solution of \eqref{eq:equ3}, where Assumption \ref{assum:fel2} holds, and let $0 \leq t\leq T$. For any $Y(0)>0$ and any $p>0$, there exist constants $M_{1,p},M_{2,p}<\infty$, such that
\begin{equation}\label{eq:equCondMp}
\mathbb{E} \left[ \sup_{u\in[0,T]}Y(u)^{p} \biggl|\mathcal{F}_{t}\right]\leq M_{1,p}(1+Y(t)^p),\quad a.s.,
\end{equation}
and
\begin{equation} \label{eq:equMp}
\mathbb{E} \left[ \sup_{u\in[0,T]}Y(u)^{p} \right] \leq M_{2,p}.
\end{equation}
\end{lemma}
\begin{proof}
The proof of \eqref{eq:equCondMp} is an application of \cite[Lemma 2.1]{2} to conditional expectations requiring an invocation of Jensen's inequality when $0<p<4$. Eq. \eqref{eq:equMp} is provided by \cite[Lemma 3.2]{7}. 
\end{proof}

Finally, we will make frequent use of the following elementary inequalities: for $n\in\mathbb{N}$ and $a_1,\ldots,a_n\in\mathbb{R}$ and $p\geq 0$,
\begin{eqnarray}
\sqrt{|a_1+a_2|}&\leq&\sqrt{|a_1|}+\sqrt{|a_2|};\label{eq:EIsqrt}\\
|a_1a_2|&\leq&\frac{1}{2}(a_1^2+a_2^2);\label{eq:EIprSp}\\
(a_1+\ldots+a_n)^p&\leq& n^{p}(|a_1|^p+\cdots+|a_n|^p).\label{eq:EIpSp}
\end{eqnarray}

\section{An Adaptive Numerical Method}\label{sec:adapt}
\cite{12} provided a framework within which to construct timestepping strategies for an adaptive explicit Euler-Maruyama numerical scheme applied to nonlinear It\^o-type SDEs of the form
\begin{equation}\label{eq:equ7}
dX(t)=f(X(t))dt+g(X(t))dB(t), \quad t \in [0,T],
\end{equation}  
over a random mesh  $\lbrace t_n \rbrace_{n \in \mathbb{N}}$ on the interval $[0,T]$ given by,
\begin{equation}\label{eq:equ8}
Y_{n+1}=Y_{n}+h_{n+1}f(Y_n)+g(Y_n)(W(t_{n+1})-W(t_n)),
\end{equation}
where $\lbrace h_n \rbrace_{n \in \mathbb{N}}$ is a sequence of random timesteps and $\lbrace t_n = \sum_{i=1}^{n} h_i \rbrace_{n=1}^N$ with $t_0=0$, so that $t_{n+1}>t_n$ for each $n$. The choice of indexing ensures consistency of notation between $t_n$ and $h_n$ and the random time step $h_{n+1}$ is determined by $Y_n$. 

Our proposed timestepping strategy will reflect dynamical considerations specific to the transformed CIR model \eqref{eq:equ3} corresponding to 
\begin{equation}\label{eq:fg}
f(y)=\frac{\alpha}{y}+\beta y,\qquad g(y)=\gamma.
\end{equation} 

\subsection{Framework for a random mesh}

\begin{definition}[\cite{LM}]\label{def:fil}
Suppose that each member of the sequence $\lbrace t_n \rbrace_{n \in \mathbb{N}}$ is an $\mathcal{F}_t$-stopping time: i.e. $\lbrace t_n \leq t \rbrace \in \mathcal{F}_t, \ \text{for all} \ t \geq 0$, where $( \mathcal{F}_t )_{t\geq 0}$ is the natural filtration of $W$. We may then define a discrete time filtration $\lbrace \mathcal{F}_{t_n} \rbrace_{n \in \mathbb{N}}$ by
\begin{equation}
\nonumber
\mathcal{F}_{t_n} = \lbrace A \in \mathcal{F} : A\cap \lbrace t_n \leq t \rbrace \in \mathcal{F}_t \rbrace, n \in \mathbb{N}.
\end{equation}
\end{definition}

\begin{assumption}\label{assum:step}
Each $h_n$ is $\mathcal{F}_{t_{n-1}}$-measurable, and $N$ is a random integer such that,
\begin{equation*}
N=\max \lbrace n \in \mathbb{N}: t_{n-1}<T \rbrace \ \text{and} \ t_N=T,
\end{equation*}
and the length of maximum and minimum stepsizes satisfy $h_{\max}=\rho  h_{\min}$, for some $1< \rho < \infty$, and 
\begin{equation}\label{eq:hmaxbound}
h_{\min} \leq h_{n} \leq h_{\max}\leq 1.
\end{equation}
\end{assumption}
\begin{remark}\label{rem:N}
The lower bound $h_{\min}$ ensures that a simulation over the interval $[0,T]$ can be completed in a finite number of timesteps, and the upper bound $h_{\max}$ prevents stepsizes from becoming too large. The latter is used as a convergence parameter in our examination of the strong convergence of the adaptive method. The random variable $N$ cannot take values outside the finite set $\{N_{\min},\ldots,N_{\max}\}$, where $N_{\min}:=\lfloor T/h_{\max}\rfloor$ and $N_{\max}:=\lceil T/h_{\min}\rceil$.
\end{remark}

$\triangle W_{n+1}:=W(t_{n+1}-W(t_n)$ is a Wiener increment over a random interval the length of which depends on $Y_n$, through which it depends on $\{W(s),\,s\in[0,t_n]\}$. Therefore $\triangle W_{n+1}$ is not independent of $\mathcal{F}_{t_n}$; indeed it is not necessarily normally distributed. Since $h_{n+1}$ is a bounded $\mathcal{F}_{t_n}$-stopping time and $\mathcal{F}_t$-measurable, then $W(t_{n+1})-W(t_n)$ is $\mathcal{F}_{t_n}$-conditionally normally distributed, by Doob's optional sampling theorem (see for example~\cite{Shiryaev96}), and for all $p\geq 2$ there exists $\Cp{p}<\infty$ such that 

\begin{eqnarray}
\mathbb{E}[W(t_{n+1})-W(t_n) |\mathcal{F}_{t_n}]&=&0,\quad a.s.;\nonumber\\
\mathbb{E}[|W(t_{n+1})-W(t_n)|^2 |\mathcal{F}_{t_n}]&=&h_{n+1},\quad a.s.;\nonumber\\
\mathbb{E}\left[\left|\int_{t_n}^{s}dW(r)\right|^{p} \middle|\mathcal{F}_{t_n}\right]&=&\Cp{p}|s-t_n|^{\frac{p}{2}} ,\quad a.s. \label{eq:equ9}
\end{eqnarray}

\subsection{Adaptive timestepping strategy}

To ensure strong convergence, our strategy is to reduce the size of each timestep if discretised solutions attempt to enter a neighbourhood of zero. If we wish to control the likelihood of invoking the backstop to avoid negative values, we will also reduce the timestep when solutions grow large.
\begin{definition}[A path-bounded time-stepping strategy]\label{def:defYn}
Let $\{Y_n\}_{n\in\mathbb{N}}$ be a solution of \eqref{eq:equ8}. We say that $\{h_n\}_{n\in\mathbb{N}}$ is a path-bounded time-stepping strategy for \eqref{eq:equ8} if the conditions of Assumption \ref{assum:step} are satisfied and there exist real non-negative constants $0\leq Q<R$ (where $R$ may be infinite if $Q\neq 0$) such that whenever $\hmin \leq h_n \leq \hmax$,
\begin{align}
Q\leq |Y_n| < R, \quad n=0,\dots, N-1. \label{eq:defYn}
\end{align}
\end{definition}

We now give two examples of path-bounded strategies that are valid for \eqref{eq:equ8}, the first with $R$ infinite (which, in conjunction with a suitable backstop method, is sufficient to ensure strong convergence), and the second with $R$ finite (which is useful if we also wish to minimise the use of the backstop to ensure positivity).
\begin{lemma}
Let $\{Y_n\}_{n\in\mathbb{N}}$ be a solution of \eqref{eq:equ8}, and let $\{h_n\}_{n\in\mathbb{N}}$ be a time-stepping strategy that satisfies Assumption \ref{assum:step}. If $\{h_n\}_{n\in\mathbb{N}}$  satisfies, for some $r\geq 1$,
\begin{equation}\label{eq:equ12}
h_{n+1}:=\max\left(h_{\min},h_{\max} \cdot \min \lbrace 1,|Y_n|^r \rbrace\right),\quad n\in\mathbb{N},
\end{equation} 
or
\begin{equation}\label{def:PBpos}
h_{n+1}=\max\left(h_{\min},h_{\max}\cdot\min(|Y_n|^r,|Y_n|^{-r})\right),\quad n\in\mathbb{N},
\end{equation}
then it is path-bounding for \eqref{eq:equ8} in the sense of Definition \ref{def:defYn}.
\end{lemma}
\begin{proof}
Suppose that \eqref{eq:equ12} holds and so $h_{n+1}\geq h_{\min}$. When $|Y_n|<1$,
\begin{equation*}
h_{n+1}\leq  h_{\max}|Y_n|^r \Leftrightarrow \frac{1}{|Y_n|^r}\leq \frac{h_{\max}}{h_{n+1}} \leq \frac{h_{\max}}{h_{\min}}= \rho,\quad n\in\mathbb{N},
\end{equation*}
and when  $|Y_n| \geq 1$ it is obvious that $\frac{1}{|Y_n|} \leq 1$, so we also have $\frac{1}{|Y_n|^r}\leq \rho < \infty$. Hence, when using the strategy defined by \eqref{eq:equ12}, 
\begin{equation*}
|Y_n|\geq\frac{1}{\rho^{1/r}},\quad \frac{1}{|Y_n|}\leq \rho^{1/r},\quad n\in\mathbb{N},
\end{equation*}
so that \eqref{eq:defYn} holds with $Q=1/\rho^{1/r}$ and $R=\infty$.

We can similarly show that \eqref{def:PBpos} is path-bounding for \eqref{eq:equ8}, with $Q=1/\rho^{1/r}$ and $R=\rho^{1/r}$.
\end{proof}
Note that for the strategies defined by \eqref{eq:equ12} and \eqref{def:PBpos}, solutions of \eqref{eq:equ8}  cannot enter the neigbourhood $(\pm 1/\rho^{1/r})$, and therefore terms of the sequence $(1/|Y_n|)_{n\in\mathbb{N}}$ are uniformly bounded from above. This has the effect of controlling inverse moments of the solutions of \eqref{eq:equ8}. Moreover for \eqref{eq:equ12}, when \eqref{eq:fg} holds,
\[
f(Y_n^2)=\frac{\alpha}{|Y_n|^2}+\beta |Y_n|^2\leq \alpha\rho^{1/r}+\beta |Y_n|^2,
\]
and therefore that strategy is admissible in the sense of \cite[Definition 2.2]{12} with $R_1=\alpha\rho^{1/r}$ and $R_2=\beta$. Similarly, \eqref{def:PBpos} is admissible with $R_1=\alpha\rho^{1/r}+\beta\rho^{-1/r}$ and $R_2=0$.

\subsection{The adaptive numerical method with backstop}

We consider an adaptive scheme based upon the following explicit Euler-Maruyama discretisation of \eqref{eq:equ3} over a random mesh given by,
\begin{equation}\label{eq:equ14}
Y_{n+1} = Y_n + h_{n+1} \left( \frac{\alpha}{Y_n}+\beta Y_n \right) + \gamma \Delta W_{n+1}.
\end{equation}
where the timestep sequence is constructed according to \eqref{eq:equ12}. For $s\in[t_n,t_{n+1})$, the continuous version is given by
\begin{equation}\label{eq:EMcns}
\widetilde Y(s)=Y_n+\int_{t_n}^s \left(\frac{\alpha}{Y_n}+\beta Y_n\right)dr+\gamma\int_{t_n}^{s}dW(r),
\end{equation}
so that $\widetilde Y(t_n)=Y_n$ for each $n\in\mathbb{N}$.

We combine this scheme with a positivity-preserving backstop scheme that is to be applied if the timestepping strategy attempts to select a timestep below $h_{\text{min}}$ (in which case we choose $h_{n+1}=h_{\text{min}}$) or if the current selected timestep and subsequently observed Brownian increment $\triangle W_{n+1}$ would result in the approximation becoming negative.
First, we define a map representing the explicit Euler scheme:
\begin{definition}\label{def:EMMap}
Define the map $\theta : \mathbb{R}^3 \rightarrow \mathbb{R} $ such that
\begin{equation}
\nonumber
\theta (y,z,h) := y+h\left(\frac{\alpha}{y}+\beta y\right)+\gamma z,
\end{equation}
so that, if $\lbrace Y_n \rbrace_{n \in \mathbb{N}}$ is defined by \eqref{eq:equ14}, then
\begin{equation}
\nonumber
Y_{n+1} = \theta(Y_n,\Delta W_{n+1}, h_{n+1}), \ n \in \mathbb{N}.
\end{equation}
\end{definition}
Next we characterise the map associated with the backstop method:
\begin{definition}\label{def:BStMap}
Define the backstop map $\varphi : \mathbb{R}^3 \rightarrow \mathbb{R} $ so that, for $s\in[t_n+h_{\min},t_n+h_{\max}]$, it satisfies 
\begin{multline}\label{eq:equ16}
\mathbb{E}\left[ \left| \varphi\left(\bar{Y}_n, \int_{t_n}^sdW(r), \int_{t_n}^sdr\right)-Y(s)\right|^2 \big|\mathcal{F}_{t_n} \right]- \big| \bar{Y}_n - Y(t_n) \big|^2\\
\leq C_1 \int_{t_n}^{s}\mathbb{E}\left[|\bar Y(r)-Y(r)|^2|\mathcal{F}_{t_n}\right]dr+C_2 |s-t_n|^{3/2}, \ n \in \mathbb{N}, \ a.s.,
\end{multline}
for some non-negative constants $C_1$ and $C_2$, independent of N, and 
\begin{equation}\label{eq:bspos}
\varphi\left(\bar{Y}_n, \int_{t_n}^sdW(r),\int_{t_n}^sdr\right)>0\quad a.s.,\quad \bar{Y}_n>0.
\end{equation}
where where $\bar Y_n:=\bar Y(t_n)$, and $\bar{Y}$ is the continuous form of our hybrid scheme constructed in next, in Definition \ref{def:map}.
\end{definition}
\begin{definition}\label{def:map}
Define the sequence of functions $\{(\bar Y(s))_{s\in[t_n,t_{n+1})}\}_{n\in\mathbb{N}}$ obeying
\begin{multline}\label{eq:equ15}
\bar{Y}(s)=\theta\left(\bar{Y}_n, \int_{t_n}^{s}dW(r), \int_{t_n}^{s}dr\right)\mathcal{I}_{\lbrace h_{\min}<h_{n+1} \leq h_{\max} \rbrace\cap\lbrace Y_{n+1}>0\rbrace} \\
+ \varphi\left(\bar{Y}_n, \int_{t_n}^{s}dW(r), \int_{t_n}^{s}dr\right)\mathcal{I}_{\lbrace h_{n+1} = h_{\min} \rbrace\cap\lbrace Y_{n+1}>0\rbrace}\\
+\varphi\left(\bar{Y}_n, \int_{t_n}^{s}dW(r), \int_{t_n}^{s}dr\right)\mathcal{I}_{\lbrace h_{\min}<h_{n+1} \leq h_{\max} \rbrace\cap\lbrace Y_{n+1}<0\rbrace},
\end{multline}
for $s\in[t_n,t_{n+1})$, $n\in\mathbb{N}$, where $\lbrace h_n \rbrace_{n \in \mathbb{N}}$ satisfies the conditions of Assumption \ref{assum:step}.
\end{definition}

In practice, rather than checking \eqref{eq:equ16} directly, we use as our backstop a method that is known to be positivity preserving and strongly convergent of order at least $1/2$. In Section \ref{sec:numerics} we use the transformed fully implicit method proposed by \cite{Alphonsi2005}; one could also choose the fully truncated method~\cite{15}.
\begin{remark}
Since the events $\{Y_{n+1}<0\}$ and $\{Y_{n+1}>0\}$ are $\mathcal{F}_{t_{n+1}}$-measurable but not $\mathcal{F}_{t_n}$-measurable, if a negative value of $Y_{n+1}$ is observed following a step of length $h_{n+1}$ we must retake the step using the backstop method, which will ensure positivity over that step by \eqref{eq:bspos}. This introduces an element of backtracking into the algorithm, but as long as the originally computed stepsize $h_{n+1}$ and Brownian increment are retained we can stay on the same trajectory while avoiding the use of a Brownian bridge. Theorem \ref{thm:pos} in Section \ref{sec:pos}, illustrated by Example \ref{ex:pos}, demonstrates that it is always possible to choose $h_{\max}$ to ensure that this particular use of the backstop can be avoided with probability $1-\varepsilon$, for arbitrarily small $\varepsilon\in(0,1)$, on each trajectory. 
\end{remark}

\section{Main Results}\label{sec:main}
In this section, we first demonstrate strong convergence of solutions of \eqref{eq:equ15} to those of \eqref{eq:equ3} under Assumption \ref{assum:fel2} and a path-bounded timestepping strategy. Second, we investigate the likelihood that the adaptive part of the method generates a negative value (triggering the use of the backstop to ensure positivity) and show how $h_{\max}$ may be chosen to control the probability of this occurring.

\subsection{Strong convergence of the adaptive method with path-bounded timestepping strategy}

\begin{lemma}\label{lem:adap}
Let $\left( Y(t) \right)_{t \in [0,T]}$ be a solution of \eqref{eq:equ3} and let $\lbrace t_n \rbrace_{n \in \mathbb{N}}$ be a random mesh such that each $t_n$ is an $\mathcal{F}_t$-stopping time. Fix $n\in\mathbb{N}$ and suppose that $t_n \leq s \leq T$, where $T\in[0,\bar T]$. 
Then, for any $1\leq p\leq 6$, we have
\begin{equation}
\nonumber
\mathbb{E} \left[ | Y(s)-Y(t_n)|^p \big|\mathcal{F}_{t_n} \right] \leq 2^{p}\gamma^p \Cp{p}|s-t_n|^{p/2} + \bar L_{n,p} |s-t_n|^p, \ a.s.,
\end{equation}
where
\begin{equation*}
\bar L_{n,p}:=2^{2p} \left( \alpha^p \frac{C(p,T)}{Y(t_n)^p} + |\beta|^p M_{1,p}(1+Y(t_n)^p)\right)
\end{equation*}
is an $\mathcal{F}_{t_n}$-measurable random variable with finite expectation, and $C(T)$, $M_{1,p}$ are the constants defined by \eqref{eq:equM-2} and \eqref{eq:equCondMp} in the statements of Lemmas \ref{lem:fin} and \ref{lem:posMB} respectively . 
\end{lemma}

\begin{proof}
Solutions of \eqref{eq:equ3} satisfy the integral equation
\begin{equation*}
Y(s)=Y(t_n)+ \int_{t_n}^{s} \left( \frac{\alpha}{Y(u)}+\beta Y(u) \right) du + \int_{t_n}^{s} \gamma dW(u),\quad t_n\leq s\leq T,
\end{equation*}
and therefore
\begin{equation*}
Y(s)-Y(t_n)=\int_{t_n}^{s} \left( \frac{\alpha}{Y(u)}+\beta Y(u) \right) du + \gamma \left( W(s)-W(t_n) \right),\quad t_n\leq s\leq T.
\end{equation*}
Using the triangle and Cauchy-Schwarz inequalities, and the elementary inequality \eqref{eq:EIpSp} with $n=2$,
\begin{eqnarray*}
\lefteqn{|Y(s)-Y(t_n)|^p}\nonumber\\
&\leq& 2^{p} \left| \int_{t_n}^{s} \left( \frac{\alpha}{Y(u)}+\beta Y(u) \right) du \right|^p + 2^{p}\gamma^p \vert W(s)- W(t_n) \vert^p \nonumber\\
&\leq& 2^{p} |s-t_n|^{p-1}\int_{t_n}^{s} \left| \frac{\alpha}{Y(u)}+\beta Y(u) \right|^p du + 2^p\gamma^p \vert W(s)- W(t_n) \vert^p\nonumber\\
&\leq& 2^{2p} |s-t_n|^{p-1} \left( \int_{t_n}^{s} \frac{\alpha^p}{Y(u)^p} du + \int_{t_n}^{s} |\beta|^p Y(u)^p du \right)\label{eq:YconSq}\\
&&\qquad\qquad\qquad\qquad\qquad\qquad\qquad\qquad+ 2^{p}\gamma^p \vert W(s)- W(t_n) \vert^p,\nonumber
\end{eqnarray*}
for $s\in[t_n,T]$. Now apply conditional expectations on both sides with respect to $\mathcal{F}_{t_n}$ and \eqref{eq:equ9} to get,
\begin{multline*}
\mathbb{E}\left[|Y(s)-Y(t_n)|^p \big|\mathcal{F}_{t_n} \right] \leq 2^{p}\gamma^p \mathbb{E}\left[\vert W(s)- W(t_n) \vert^p \big|\mathcal{F}_{t_n}\right]\\+2^{2p} |s-t_n|^{p-1} \left(\mathbb{E}\left[ \int_{t_n}^{s} \frac{\alpha^p}{Y(u)^p} du \bigg| \mathcal{F}_{t_n} \right]+\mathbb{E}\left[ \int_{t_n}^{s} |\beta|^p Y(u)^p du \bigg| \mathcal{F}_{t_n} \right] \right)\\
\leq 2^{p}\gamma^p\Cp{p}|s-t_n|^{p/2}+2^{2p} |s-t_n|^{p-1}\left( \alpha^p \int_{t_n}^{s} \frac{C(p,T)}{Y(t_n)^p}du\right.\\+\left.|\beta|^p \int_{t_n}^{s} M_{1,p}(1+Y(t_n)^p) du\right), \quad a.s,
\end{multline*}
where we have used \eqref{eq:equM-2} and \eqref{eq:equCondMp} from the statement of Lemma \ref{lem:fin} at the last step. Therefore
\begin{multline*}
\mathbb{E}\left[|Y(s)-Y(t_n)|^p \big|\mathcal{F}_{t_n} \right]\leq 2^{p}\gamma^p\Cp{p}|s-t_n|^{p/2}\\+2^{2p}\left(\alpha^p\frac{C(p,T)}{Y(t_n)^p}+\beta M_{1,p}(1+Y(t_n)^p)\right)|s-t_n|^p, \ a.s,
\end{multline*}
as required.

\end{proof}

\begin{lemma}\label{lem:tay}
Let $\left( Y(t)\right)_{t \in [0,T]}$ be a solution of \eqref{eq:equ3} and let Assumption \ref{assum:fel2} hold. Let $\lbrace t_n \rbrace_{n \in \mathbb{N}}$ arise from  the adaptive timestepping strategy satisfying \eqref{eq:defYn} in Definition \ref{def:defYn} for some $0<Q<R$ , and formulate the Taylor expansion of $f(Y(s))$ around $Y(t_n)$, where $f$ is as given in \eqref{eq:fg}, as
\begin{equation}\label{eq:equ18}
f(Y(s))=f(Y(t_n))+R_f(s,t_n,Y(t_n)),\quad s\in[t_n,t_{n+1}],
\end{equation}
where 
\begin{multline}\label{eq:Rf}
R_f(s,t_n,Y(t_n)) = \int_{0}^{1} Df\left(Y(t_n)+ \tau (Y(s)-Y(t_n))\right)(Y(s)-Y(t_n))d \tau.
\end{multline}
For any $1\leq p\leq 3$, the $p^{th}$ conditional moment of $R_f(s,t_n,Y(t_n))$ satisfies 
\begin{equation}\label{eq:Rfboundp}
\mathbb{E} \left[ |R_f|^p \big| \mathcal{F}_{t_n} \right]\leq K_{n,p}h_{n+1}^{p/2},\quad a.s.,
\end{equation}
where $K_{n,p}$ is an a.s. finite and $\mathcal{F}_{t_n}$-measurable random variable given by 
\begin{multline*}
K_{n,p}=2^{p}|\beta|^p \left(2^{p}\gamma^p\Cp{p}+\bar L_{n,p}h_{n+1}^{p/2}\right)\\
+2^{p}\frac{\alpha^p\sqrt{C(p,T)}}{Y(t_n)^{2p}}\left(2^{p}\gamma^p\Cp{2p}^{1/2}+\bar L_{n,2p}^{1/2}h_{n+1}^{p/2}\right).
\end{multline*}
Moreover, there exists $K_p$ independent of $n$ such that
\begin{equation}\label{eq:finiteExpCoeff}
K_p:=\mathbb{E}[K_{n,p}]<\infty.
\end{equation}
\end{lemma}

\begin{proof}
By direct substitution of $f(y)$ from \eqref{eq:fg} into \eqref{eq:Rf}, evaluating the integral in $\tau$, and taking the $p^{th}$-moment conditional upon $\mathcal{F}_{t_n}$, we get
\begin{equation*}
\mathbb{E}\left[|R_f|^p \big|\mathcal{F}_{t_n}\right] = \mathbb{E}\left[ \left| (Y(s)-Y(t_n)) \left( \beta - \frac{\alpha}{Y(s)Y(t_n)} \right) \right|^p \bigg| \mathcal{F}_{t_n} \right].
\end{equation*}
Using the triangle inequality and \eqref{eq:EIpSp} we get
\begin{multline}\label{eq:optnotopt}
\mathbb{E}\left[|R_f|^p \big|\mathcal{F}_{t_n} \right]\leq 2^{p}|\beta|^p\mathbb{E}\left[ \left|Y(s)-Y(t_n) \right|^p \bigg| \mathcal{F}_{t_n} \right] \\
\qquad \qquad \qquad + \frac{2^{p}\alpha^p}{Y(t_n)^p}\mathbb{E}\left[ \left|\left( (Y(s)-Y(t_n))\cdot\frac{1}{Y(s)} \right) \right|^p \bigg| \mathcal{F}_{t_n} \right]. 
\end{multline}

Next apply Lemma \ref{lem:adap} followed by the Cauchy-Schwarz inequality to get
\begin{eqnarray*}
\mathbb{E}[|R_f|^p \big|\mathcal{F}_{t_n} ]&\leq&2^{p} |\beta|^p h_{n+1}^{p/2}\left(2^{p}\gamma^p\Cp{p} + \bar L_{n,p} h_{n+1}^{p/2}\right)\\
&&\,\,+2^{p}\frac{\alpha^p}{Y(t_n)^p}\sqrt{\mathbb{E}[|Y(s)-Y(t_n)|^{2p}|\mathcal{F}_{t_n}]}\sqrt{\mathbb{E}\left[\frac{1}{|Y(s)|^{2p}}\bigg|\mathcal{F}_{t_n}\right]}\\
&\leq& 2^{p}|\beta|^p h_{n+1}^{p/2}\left(2^{p}\gamma^p\Cp{p} + \bar L_{n,p} h_{n+1}^{p/2}\right)\\
&&\,\,+2^{p}\frac{\alpha^p\sqrt{C(p,T)}}{Y(t_n)^{2p}}\sqrt{\mathbb{E}[|Y(s)-Y(t_n)|^{2p}|\mathcal{F}_{t_n}]}.
\end{eqnarray*}

Again applying Lemma \ref{lem:adap} and the elementary inequality \eqref{eq:EIsqrt} this becomes
\begin{eqnarray*}
\lefteqn{\mathbb{E}[|R_f|^p \big|\mathcal{F}_{t_n}]\leq 2^{p}|\beta|^p h_{n+1}^{p/2}\left(2^{p}\gamma^p\Cp{p}+\bar L_{n,p}h_{n+1}^{p/2}\right)}
\\
&&\qquad \qquad \qquad \qquad +2^{p}h_{n+1}^{p/2}\frac{\alpha^p\sqrt{C(p,T)}}{Y(t_n)^{2p}}\left(2^{p}\gamma^p\Cp{2p}^{1/2}+\bar L_{n,2p}^{1/2}h_{n+1}^{p/2}\right),
\end{eqnarray*}
from which the statement of the Lemma follows when we observe that the a.s. finiteness of $K_{n,p}$ is ensured by Assumption \ref{assum:fel2}, and \eqref{eq:finiteExpCoeff} is ensured by Lemmas \ref{lem:fin} \& \ref{lem:posMB}. 
\end{proof}
\begin{remark}
It is also possible to estimate the second expectation in  \eqref{eq:optnotopt} by an application of It\^o's formula, rather than the Cauchy-Schwarz inequality. However, this increases the maximum number of finite inverse moments of $Y$ required from $2p$ to $3p$ and does not improve the order of the bound \eqref{eq:Rfboundp}.
\end{remark}

\begin{lemma}\label{lem:back}
Let $\left( Y(t_n) \right)_{t_n \in [0,T]}$ be the solution of \eqref{eq:equ3} with initial value $Y(0)=Y_0=\sqrt{X_0}$. Let $\left(\widetilde Y(s)\right)_{s \in [t_n,t_{n+1}]}$ be a solution of \eqref{eq:EMcns} over the interval $[t_n,t_{n+1}]$ and $\lbrace h_n \rbrace_{n \in \mathbb{N}}$ be a sequence of random timesteps  defined by \eqref{eq:equ3} and $\lbrace t_n = \sum_{i=1}^{n} h_i \rbrace_{n=1}^N$ with $t_0=0$. Then for $n,p\in\mathbb{N}$, there exists an $\mathcal{F}_{t_n}$-measurable random variable $\bar{K}_n$ with finite expectation $K_n:=\mathbb{E}[\bar{K}_n]<\infty$ such that
\begin{equation}
\mathbb{E} \left[ E(t_{n+1})^2 \big| \mathcal{F}_{t_n} \right] -E(t_{n})^2 \leq \int_{t_n}^{t_{n+1}}\mathbb{E}\left[E(r)^2|\mathcal{F}_{t_n}\right]+\bar{K}_nh_{n+1}^2,\quad a.s.,\label{eq:main2}
\end{equation}
where the error $E(s):=Y(s)-\widetilde Y(s)$, $s\in[t_n,t_{n+1}]$.
\end{lemma}
\begin{proof}

For $s\geq t_n$, we subtract \eqref{eq:EMcns} from \eqref{CIRtInt} to get
\begin{eqnarray}
E(s)&=& Y(s)-\widetilde{Y}(s) \nonumber \\
&=& \left[ Y(t_n)+ \int_{t_n}^{s} f(Y(r)) dr + \gamma \int_{t_n}^{s} dW(r) \right] \nonumber\\
& & \qquad \qquad  \qquad  - \left[ Y_n+ \int_{t_n}^{s} f(Y_n) dr  + \gamma \int_{t_n}^{s} dW(r) \right] \nonumber \\
&=& E(t_n)+\int_{t_n}^{s} \tilde{f}(Y(r),Y_n)dr, \label{eq:E}
\end{eqnarray}
where $f$ is defined as in \eqref{eq:fg} and $\tilde{f}(Y(r),Y_n)=f(Y(r))-f(Y_n)$. Applying the It\^o formula and setting $s=t_{n+1}$, we can write,
\begin{equation*}
E(t_{n+1})^2=E(t_n)^2+2\int_{t_n}^{t_{n+1}}E(r)\tilde{f}(Y(r),Y_n)dr.
\end{equation*}

By \eqref{eq:equ18} in the statement of Lemma \ref{lem:tay}, $$\tilde{f}(Y(r),Y_n)=\tilde{f}(Y(t_n),Y_n)+R_f(r,t_n,Y(t_n),$$ where $R_f$ is defined in \eqref{eq:Rf}. This, and an application of \eqref{eq:EIprSp}  gives
\begin{eqnarray}
\lefteqn{E(t_{n+1})^2-E(t_n)^2}\nonumber\\ &=&2\int_{t_n}^{t_{n+1}}E(r)R_f(r,t_n,Y(t_n))dr+2\int_{t_n}^{t_{n+1}}E(r)\tilde{f}(Y(t_n),Y_n)dr\nonumber\\
&\leq&\int_{t_n}^{t_{n+1}}E(r)^2dr+\int_{t_n}^{t_{n+1}}R_f(r,t_n,Y(t_n))^2dr\nonumber\\
&&\qquad\qquad\qquad\qquad\qquad\qquad+2\int_{t_n}^{t_{n+1}}E(r)\tilde{f}(Y(t_n),Y_n)dr.\label{eq:lasttry}
\end{eqnarray}
Consider the third term on the RHS of \eqref{eq:lasttry}, and substitute \eqref{eq:E}  into the integrand:
\begin{eqnarray}
\lefteqn{\int_{t_n}^{t_{n+1}}E(r)\tilde{f}(Y(t_n),Y_n)dr}\nonumber\\
&=&E(t_n)\tilde{f}(Y(t_n),Y_n)\int_{t_n}^{t_{n+1}}dr+\tilde{f}(Y(t_n),Y_n)^2\int_{t_n}^{t_{n+1}}\int_{t_n}^{r}du\,dr\nonumber\\
&&\qquad \qquad +\tilde{f}(Y(t_n),Y_n)\int_{t_n}^{t_{n+1}}\int_{t_n}^{r}R_f(u,t_n,Y(t_n))du\,dr\nonumber\\
&\leq&\beta\int_{t_n}^{t_{n+1}}E(t_n)^2dr+\tilde{f}(Y(t_n),Y_n)^2h_{n+1}^2\nonumber\\
&&\qquad \qquad+\tilde{f}(Y(t_n),Y_n)\int_{t_n}^{t_{n+1}}\int_{t_n}^{r}R_f(u,t_n,Y(t_n))du\,dr\nonumber\\
&\leq&\tilde{f}(Y(t_n),Y_n)^2h_{n+1}^2\nonumber\\
&&\qquad\qquad+\tilde{f}(Y(t_n),Y_n)\int_{t_n}^{t_{n+1}}\int_{t_n}^{r}R_f(u,t_n,Y(t_n))du\,dr.\label{eq:presub}
\end{eqnarray}

Now substitute \eqref{eq:presub} into \eqref{eq:lasttry}, to get
\begin{multline}\label{eq:errorpath}
E(t_{n+1})^2-E(t_n)^2\leq \int_{t_n}^{t_{n+1}}E(r)^2dr+\int_{t_n}^{t_{n+1}}R_f(r,t_n,Y(t_n))^2dr\\+2\tilde{f}(Y(t_n),Y_n)^2h_{n+1}^2+2\tilde{f}(Y(t_n),Y_n)\int_{t_n}^{t_{n+1}}\int_{t_n}^{r}R_f(u,t_n,Y(t_n))du\,dr.
\end{multline}
Apply expectations to both sides of \eqref{eq:errorpath}, conditional upon $\mathcal{F}_{t_n}$, to get
\begin{multline*}
\mathbb{E}\left[E(t_{n+1})^2|\mathcal{F}_{t_n}\right]-E(t_n)^2\leq \int_{t_n}^{t_{n+1}}\mathbb{E}[E(r)^2|\mathcal{F}_{t_n}]dr\\+\int_{t_n}^{t_{n+1}}\mathbb{E}[R_f(r,t_n,Y(t_n))^2|\mathcal{F}_{t_n}]dr
+2\tilde{f}(Y(t_n),Y_n)^2h_{n+1}^2\\+2\tilde{f}(Y(t_n),Y_n)\int_{t_n}^{t_{n+1}}\int_{t_n}^{r}\mathbb{E}[R_f(u,t_n,Y(t_n))|\mathcal{F}_{t_n}]du\,dr,\quad a.s.
\end{multline*}
Apply the bound \eqref{eq:Rfboundp} in the statement of Lemma \ref{lem:tay} with $p=1,2$ to get
\begin{multline*}
\mathbb{E}\left[E(t_{n+1})^2|\mathcal{F}_{t_n}\right]-E(t_n)^2\leq \int_{t_n}^{t_{n+1}}\mathbb{E}[E(r)^2|\mathcal{F}_{t_n}]dr\\+(2\tilde{f}(Y(t_n),Y_n)+K_{n,2})h_{n+1}^2+2\tilde{f}(Y(t_n),Y_n)K_{n,1}h_{n+1}^{5/2},\quad a.s.
\end{multline*}
Since, by \eqref{eq:hmaxbound} in the statement of Assumption \ref{assum:step}, $h_{n+1}\leq h_{\max}\leq 1$, the statement of the Lemma now follows, with $\bar{K}_n:=2\tilde{f}(Y(t_n),Y_n)+K_{n,2}+2\tilde{f}(Y(t_n),Y_n)K_{n,1}$.\\

To see that $\mathbb{E}[\bar{K}_n]<\infty$, note that by \eqref{eq:defYn} in Definition \ref{def:defYn}, and since $\beta<0$, $f(Y_n)\leq \alpha/Q$. The finiteness of $\mathbb{E}[K_{n,2}]$ and $\mathbb{E}[K_{n,1}]$ is given by \eqref{eq:equM-2} in the statement of Lemma \ref{lem:fin} with $p=6,3$ respectively, along with \eqref{eq:equCondMp} in the statement of Lemma \ref{lem:posMB}.

\end{proof}

\begin{theorem}\label{thm:scc}
Let $\left( Y(t) \right)_{t \in [0,T]}$ be the solution of \eqref{eq:equ3} with initial value $Y(0)=Y_0=\sqrt{X_0}$, and suppose that Assumption \ref{assum:fel2} holds. Let $(\bar Y(t))_{t\in[0,T]}$ be a solution of \eqref{eq:equ15} with initial value $\bar Y(0)=Y(0)$ and path-bounded timestepping strategy $\lbrace h_n \rbrace_{n \in \mathbb{N}}$ satisfying the conditions of Definition \ref{def:defYn} for some $0<Q<R$, with $R$ possibly infinite. There exists $C >0$, independent of $h_{\max}$, such that
\begin{equation}
\nonumber
\mathbb{E}[|Y(T)-\bar{Y}(T)|^2]\leq Ch_{\max}.
\end{equation}

\end{theorem}

\begin{proof}

From \eqref{eq:main2} in the statement of Lemma \ref{lem:back}, we have that when $h_{n+1}\geq h_{\min}$,
\begin{equation}\label{eq:mainEE}
\mathbb{E} \left[ E(t_{n+1})^2 | \mathcal{F}_{t_n} \right] - E(t_{n})^2 \leq \int_{t_n}^{t_{n+1}}\mathbb{E} \left[ E(r)^{2}| \mathcal{F}_{t_n} \right]dr + \bar{K}_nh_{n+1}^2,\quad a.s.
\end{equation}
Suppose that $h_{n+1} < h_{\min}$ and $Y_{n+1}$ is generated from $Y_n$ via an application of the backstop method over a single step of length $h_{\min}$. This corresponds to single application of the map $\varphi$ in Definition \ref{def:map} and therefore the relation  \eqref{eq:equ16} holds.

\Note{We now combine \eqref{eq:equ16} and \eqref{eq:mainEE} to generate a single one-step error estimate for the hybrid method given by \eqref{eq:equ15}.} Define the positive constant $\Gamma_1=C_1\vee 1$ and $\mathcal{F}_{t_n}$-measurable random variable $\bar\Gamma_{n,2}=C_2\vee \bar{K}_n$, for $n\in\mathbb{N}$. Noting again that, by \eqref{eq:hmaxbound} in the statement of Assumption \ref{assum:step}, $h_{n+1}\leq h_{\text{max}}\leq 1$, we see that 
\eqref{eq:equ15} satisfies, on almost all trajectories, 
\begin{equation}\label{eq:onestep}
\mathbb{E}\left[E(t_{n+1})^2|\mathcal{F}_{t_n}\right]-E(t_n)^2
\leq \Gamma_1\int_{t_n}^{t_{n+1}}\mathbb{E}\left[E(r)^2|\mathcal{F}_{t_n}\right]dr+ \bar\Gamma_{n,2}h_{n+1}^{2}.
\end{equation}

Sum both sides of \eqref{eq:onestep} over $n=0,\ldots,N-1$ and take expectations:
\begin{multline}\label{eq:onestepSummed}
\expect{\sum_{n=0}^{N-1}\left(\mathbb{E}\left[E(t_{n+1})^2|\mathcal{F}_{t_n}\right]-E(t_n)^2\right)}
\\
\leq \Gamma_1\expect{\sum_{n=0}^{N-1}\int_{t_n}^{t_{n+1}}\mathbb{E}\left[E(r)^2|\mathcal{F}_{t_n}\right]dr}+\expect{\sum_{n=0}^{N-1} \bar\Gamma_{n,2}h_{n+1}^{2}}.
\end{multline}

Consider first the LHS of \eqref{eq:onestepSummed}. Since $N$ is a $\mathcal{F}_{t_n}$-stopping time, the event $\{N\leq n\}\in\mathcal{F}_{t_n}$. Moreover $N\leq N_{\max}$. So we can write
\begin{eqnarray}
\lefteqn{\expect{\sum_{n=0}^{N-1}\left(\mathbb{E}\left[E(t_{n+1})^2|\mathcal{F}_{t_n}\right]-E(t_n)^2\right)}}\nonumber\\
&=&\sum_{n=0}^{N_{\max}-1}\expect{\left(\mathbb{E}\left[E(t_{n+1})^2|\mathcal{F}_{t_n}\right]-E(t_n)^2\right)\mathcal{I}_{\{N\geq n+1\}}}\nonumber\\
&=&\sum_{n=0}^{N_{\max}-1}\left(\expect{E(t_{n+1})^2\mathcal{I}_{\{N\geq n+1\}}|\mathcal{F}_{t_n}}-E(t_n)^2\mathcal{I}_{\{N\geq n+1\}}\right)\nonumber\\
&=&\sum_{n=0}^{N_{\max}-1}\left(\expect{E(t_{n+1})^2\mathcal{I}_{\{N\geq n+1\}}}-\expect{E(t_n)^2\mathcal{I}_{\{N\geq n+1\}}}\right)\nonumber\\
&=&\expect{E(t_N)^2}=\expect{E(T)^2}.\label{eq:LHS}
\end{eqnarray}

Similarly, the RHS of \eqref{eq:onestepSummed} can be written
\begin{multline*}
\Gamma_1\expect{\sum_{n=0}^{N-1}\int_{t_n}^{t_{n+1}}\mathbb{E}\left[E(r)^2|\mathcal{F}_{t_n}\right]dr+\sum_{n=0}^{N-1} \bar\Gamma_{n,2}h_{n+1}^{2}}\\
=\Gamma_1\underbrace{\expect{\sum_{n=0}^{N_{\max}-1}\int_{t_n}^{t_{n+1}}\mathbb{E}\left[E(r)^2|\mathcal{F}_{t_n}\right]\mathcal{I}_{\{N\geq n+1\}}dr}}_{(I)}\\
+\underbrace{\expect{\sum_{n=0}^{N_{\max}-1} \bar\Gamma_{n,2}h_{n+1}^{2}\mathcal{I}_{\{N\geq n+1}\}}}_{(II)}.
\end{multline*}

Consider first (I). Since $t_n$, $t_{n+1}$, and $\mathcal{I}_{\{N\geq n+1\}}$ are $\mathcal{F}_{t_n}$-measurable we can bring $\mathcal{I}_{\{N\geq n+1\}}$ inside the conditional expectation, and exchange the order of integration and conditional expectation as follows
\begin{eqnarray}
\lefteqn{\expect{\sum_{n=0}^{N_{\max}-1}\int_{t_n}^{t_{n+1}}\mathbb{E}\left[E(r)^2|\mathcal{F}_{t_n}\right]\mathcal{I}_{\{N\geq n+1\}}dr}}\nonumber\\
&=&\sum_{n=0}^{N_{\max}-1}\expect{\int_{t_n}^{t_{n+1}}\mathbb{E}\left[E(r)^2|\mathcal{F}_{t_n}\right]\mathcal{I}_{\{N\geq n+1\}}dr}\nonumber\\
&=&\sum_{n=0}^{N_{\max}-1}\expect{\int_{t_n}^{t_{n+1}}\mathbb{E}\left[E(r)^2\mathcal{I}_{\{N\geq n+1\}}|\mathcal{F}_{t_n}\right]dr}\nonumber\\
&=&\sum_{n=0}^{N_{\max}-1}\expect{\mathbb{E}\left[\int_{t_n}^{t_{n+1}}E(r)^2\mathcal{I}_{\{N\geq n+1\}}dr\bigg|\mathcal{F}_{t_n}\right]}\nonumber\\
&=&\sum_{n=0}^{N_{\max}-1}\expect{\int_{t_n}^{t_{n+1}}E(r)^2\mathcal{I}_{\{N\geq n+1\}}dr}\nonumber\\
&=&\expect{\sum_{n=0}^{N_{\max}-1}\int_{t_n}^{t_{n+1}}E(r)^2\mathcal{I}_{\{N\geq n+1\}}dr}\nonumber\\
&=&\expect{\sum_{n=0}^{N_{\max}-1}\int_{t_n}^{t_{n+1}}E(r)^2\mathcal{I}_{\{r\leq T\}}dr}\nonumber\\
&=&\expect{\int_{0}^{T}E(r)^2dr}=\int_{0}^{T}\expect{E(r)^2}dr.\label{eq:(I)}
\end{eqnarray}

Finally consider (II). We have, since $\expect{\bar\Gamma_{n,2}\mathcal{I}_{\{N\geq n+1\}}}\leq\expect{\bar\Gamma_{n,2}}\leq C_2\vee K_2=:\Gamma_2$ for all $n=0,\ldots,N_{\max}-1$, $N_{\max}=\lceil T/h_{\min}\rceil$, and by Assumption \ref{eq:hmaxbound}, $\rho h_{\min}=h_{\max}\leq 1$,
\begin{eqnarray}
\expect{\sum_{n=0}^{N_{\max}-1} \bar\Gamma_{n,2}h_{n+1}^{2}\mathcal{I}_{\{N\geq n+1}\}}&\leq&\expect{ h_{\max}^2\sum_{n=0}^{N_{\max}-1}\bar\Gamma_{n,2}\mathcal{I}_{\{N\geq n+1\}}}\nonumber\\
&=& h_{\max}^2\sum_{n=0}^{N_{\max}-1}\expect{\bar\Gamma_{n,2}\mathcal{I}_{\{N\geq n+1\}}}\nonumber\\
&\leq& h_{\max}^2\Gamma_2N_{\max}\leq h_{\max}^2\Gamma_2\left(\frac{T}{h_{\min}}+1\right)\nonumber\\
&\leq&(\rho T+1)\Gamma_2 h_{\max}.\label{eq:(II)}
\end{eqnarray}

Substituting \eqref{eq:LHS}, \eqref{eq:(I)}, and \eqref{eq:(II)} back into \eqref{eq:onestepSummed} we get
\[
\mathbb{E}[E(T)^2]\leq\Gamma_1\int_{0}^{T}\mathbb{E}[E(r)^2]dr+(\rho T+1)\Gamma_2h_{\max}.
\]
Since this inequality holds if $T$ is varied continuously over $[0,\bar T]$, for any $\bar T<\infty$ (see \cite{13} for a demonstration)
an application of Gronwall's inequality gives the result.
\end{proof}

\subsection{On the positivity of an adaptive method with path-bounded strategy}\label{sec:pos}

In this section we assume Feller's condition ($2\kappa\lambda\geq \sigma^2$) to ensure that solutions of \eqref{eq:equ1} remain a.s. positive, but we do not require that Assumption \ref{assum:fel2} holds. 

\subsubsection{Probability of positivity over a single step}
Consider the timestepping strategy defined by \eqref{eq:equ12} with $r=1$, satisfying Definition \ref{def:defYn} with $R=\infty$. The probability of solutions of \eqref{eq:equ14} becoming negative after a single step with this strategy, and hence triggering a use of the backstop method, is given by
\begin{equation*}
\mathbb{P}\left[Y_{k+1}<0|Y_k=y>0\right]=\Phi\left(a(y)\right),\quad y>h_{\min},
\end{equation*}
where $\Phi(x)=\frac{1}{\sqrt{2\pi}}\int_{-\infty}^{x}e^{-s^2/2}ds$, and
\[
a(y)=\frac{-y-\left(\frac{\alpha}{y}+\beta y\right)h_{\max}(1\wedge y^r)}{\gamma\sqrt{h_{\max}(1\wedge y^r)}}.
\]
Figure \ref{fig:plotsOneStepProb} presents two surface plots of these one-step probabilities against $y$ and $h_{\max}$. Feller's condition is satisfied in both cases. In Figure \ref{fig:plotsOneStepProb} (a), when Assumption \ref{assum:fel2} holds, we see that the probability of invoking the backstop to avoid a negative value is highest when $h_{\max}$ is large and $Y_n$ is close to or above $1$, in which case the timestepping strategy will tend to select $h_{n+1}$ to be close to $h_{\max}$. This probability drops off rapidly as $h_{\max}$ reduces. 
In Figure \ref{fig:plotsOneStepProb} (b), when Assumption \ref{assum:fel2} does not hold, the highest probabilities of invoking the backstop for preserving positivity when $Y_n$ is close to $h_{\min}$. However the maximum probability is significantly lower than any seen in Figure \ref{fig:plotsOneStepProb} (a).

\begin{figure}
\begin{center}
$\begin{array}{@{\hspace{-0.1in}}c@{\hspace{-0.25in}}c}
{\small \kappa\lambda>2\sigma^2} & {\small \kappa\lambda<2\sigma^2}\\
\scalebox{0.45}{\includegraphics{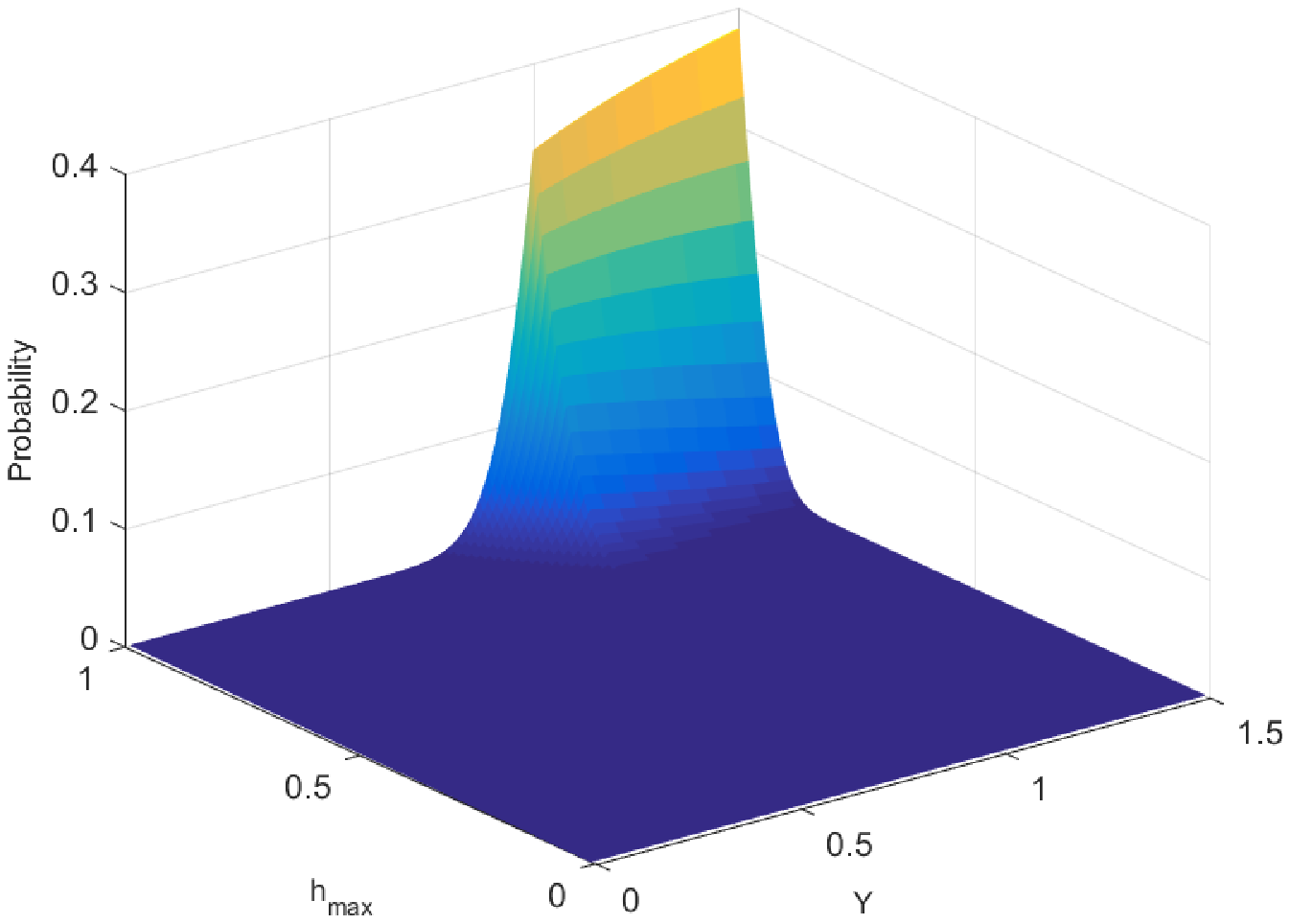}} & \scalebox{0.45}{\includegraphics{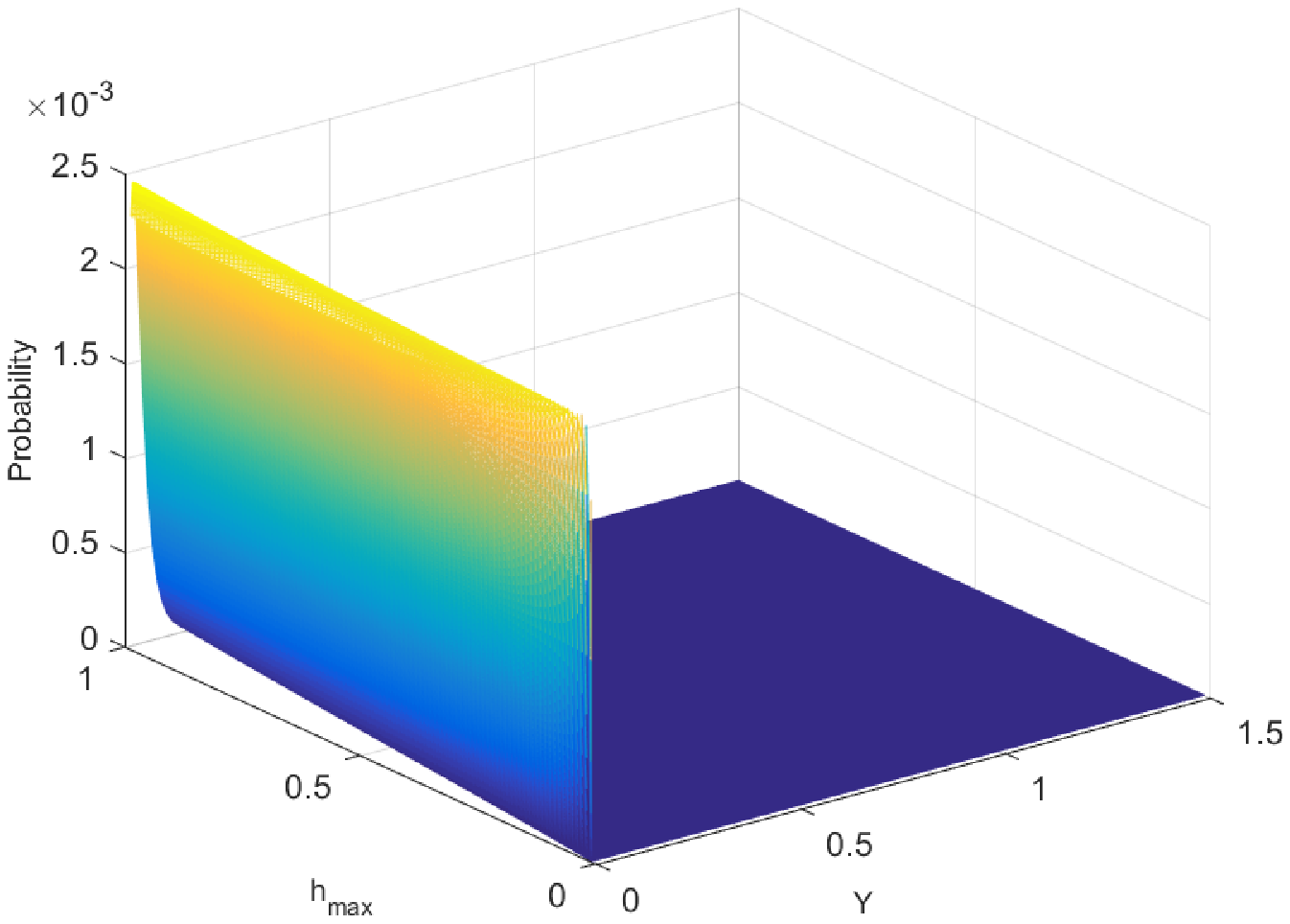}}\\
\mbox{\bf\small (a): $\sigma=0.2$, $\lambda=0.05$, $\kappa=2$} & \mbox{\bf\small (b): $\sigma=0.2$, $\lambda=0.05$, $\kappa=1$}\\
\end{array}$
\end{center}
\caption{Surface plots of probabilities of solutions of \eqref{eq:equ8} with \eqref{eq:fg} and timestepping strategy \eqref{eq:equ12} becoming negative over a single step, for $h_{\max}\in[0.01,1]$ and $Y=y\in[h_{\min},1.5]$ representing the value of the solution before taking the step. We take $\rho=2^6$. In {\bf (a)} Assumption \ref{assum:fel2} holds. In {\bf (b)} Assumption \ref{assum:fel2} does not hold.} \label{fig:plotsOneStepProb}
\end{figure}

\subsubsection{Probability of positivity over a full trajectory}
If we require path-bounded strategies where $R<\infty$, it is possible to derive an upper limit on $h_{\max}$ that is sufficient, over the entire trajectory, to keep the probability of needing the backstop scheme to prevent a negative value below some arbitrarily small tolerance. Our analysis reworks and extends the approach taken in the proof of \cite[Theorem 4.3]{KRR}, using adaptive timestepping to handle unboundedness in the drift term. By allowing the use of the backstop to ensure a minimum timestep, along with an application of the Law of Total Probability, we can avoid fixing the random number of steps $N$.

\begin{theorem}\label{thm:pos}
Let $\{Y_n\}_{n=0}^{N}$ be a solution of \eqref{eq:equ15}, with initial value $Y_0>0$, evaluated on a random mesh $\{h_n\}_{n=1}^N$ satisfying the conditions of Definition \ref{def:defYn} with $R<\infty$. Suppose also that $Y_0\in(0,R)$. Then, for each $\varepsilon\in(0,1)$ there exists $\bar h_{\max}(\varepsilon)>0$ such that, for all $h_{\max}\in(0,\bar h_{\max}(\varepsilon))$

 \[
 \mathbb{P}[\mathcal{R}_{N}]>1-\varepsilon,
 \]
where $\mathcal{R}_N:=\bigcap_{j=0}^{N}\{Y_j>0\}$.
\end{theorem} 
\begin{proof}

Since, by \eqref{eq:bspos}, the backstop method will ensure positivity over a single step if $h_{n+1}=h_{\min}$, the event $\{Y_{n+1}>0\}$ is equivalent to the following union:

\begin{eqnarray*}
\left\{\left\{\frac{\triangle W_{n+1}}{\sqrt{h_{n+1}}}>-\frac{1}{\gamma}\left(\frac{Y_n}{h_{n+1}}+\frac{\alpha\sqrt{h_{n+1}}}{Y_n}+\beta\sqrt{h_{n+1}}Y_n\right)\right\}\cap\{h_{n+1}>h_{\min}\}\right\}\\
\cup\left\{h_{n+1}=h_{\min}\right\}.
\end{eqnarray*}
Moreover, when $Y_n>0$ satisfies \eqref{eq:defYn} in Definition \ref{def:defYn},
\begin{multline}\label{eq:smallEv}
\left\{u\in-\frac{1}{\gamma}\left(\frac{Q}{h_{\max}}+\frac{\alpha\sqrt{h_{\min}}}{R}+\beta\sqrt{h_{\max}}R\right)\right\}\\
\subseteq\left\{ u\in-\frac{1}{\gamma}\left(\frac{Y_n}{h_{n+1}}+\frac{\alpha\sqrt{h_{n+1}}}{Y_n}+\beta\sqrt{h_{n+1}}Y_n\right)\right\}.
\end{multline}

For each $i=N_{\min},\ldots N_{\max}$, define  $\Omega_i:=\{\omega\in\Omega\,:\,N(\omega)=i\}$, so that $\{\Omega\}_{i=N_{\min}}^{N_{\max}}$ is a finite partition of the sample space $\Omega$. On each $\Omega_i$ define the sequence of sub-events
\[
\mathcal{R}_n(i)=\{Y_n>0,Y_{n-1}>0,\ldots,Y_1>0,Y_0>0\}\cap\Omega_i,\quad n=0,1,\ldots i.
\]

Recall that the random variable $\triangle W_{n+1}/\sqrt{h_{n+1}}$ is distributed conditionally upon $\mathcal{F}_{t_n}$ like a standard Normal random variable. Moreover, $\mathcal{R}_n(i)\in\mathcal{F}_{t_n}$ for $n=0,\ldots,i$ and $i=N_{\min},\ldots,N_{\max}$. Let $\Phi$ denote the distribution function of a standard Normal random variable, and suppose $\{\xi_{n}\}_{n\in\mathbb{N}}$ is a sequence of mutually independent standard normal random variables. 
Then
\begin{eqnarray*}
\lefteqn{\mathbb{P}[Y_{n+1}>0|\mathcal{R}_{n}(i)]} \nonumber \\
&=&\frac{\mathbb{P}[\{Y_{n+1}>0\}\cap\{\mathcal{R}_n(i)\}]}{\mathbb{P}[\mathcal{R}_n(i)]}=\frac{\expect{\expect{I_{\{Y_{n+1}>0\}\cap\{\mathcal{R}_n(i)\}}|\mathcal{F}_{t_n}}}}{\mathbb{P}[\mathcal{R}_{n}(i)]}\\
&\geq&\expect{\expect{I_{\{Y_{n+1}>0\}\cap\{\mathcal{R}_n(i)\}}|\mathcal{F}_{t_n}}}=\expect{\mathbb{P}[\{Y_{n+1}>0\}\cap\{\mathcal{R}_n(i)\}|\mathcal{F}_{t_n}]}\\
&\geq&\expect{\mathbb{P}[h_{n+1}=h_{\min}|\mathcal{F}_{t_n}]}\\
&&\qquad  +\expect{\mathbb{P}\left[\frac{\triangle W_{n+1}}{\sqrt{h_{n+1}}}>-\frac{1}{\gamma}\left(\frac{Q}{h_{\max}}+\frac{\alpha\sqrt{h_{\min}}}{R}+\beta\sqrt{h_{\max}}R\right)\bigg|\mathcal{F}_{t_n}\right]}\\
&\geq&\expect{\mathbb{P}\left[\frac{\triangle W_{n+1}}{\sqrt{h_{n+1}}}>-\frac{1}{\gamma}\left(\frac{Q}{h_{\max}}+\frac{\alpha\sqrt{h_{\min}}}{R}+\beta\sqrt{h_{\max}}R\right)\bigg|\mathcal{F}_{t_n}\right]}\\
&=&\mathbb{P}\left[\xi_{n+1}>-\frac{1}{\gamma}\left(\frac{Q}{h_{\max}}+\frac{\alpha\sqrt{h_{\max}}}{R\sqrt{\rho}}+\beta\sqrt{h_{\max}}R\right)\right]\\
&=&1-\Phi\left(-\frac{1}{\gamma}\left(\frac{Q}{h_{\max}}+\sqrt{h_{\max}}\left(\frac{\alpha}{R\sqrt{\rho}}+\beta R\right)\right)\right)\\
&=&\Phi\left(\frac{1}{\gamma}\left(\frac{Q}{h_{\max}}+\sqrt{h_{\max}}\left(\frac{\alpha}{R\sqrt{\rho}}+\beta R\right)\right)\right),\quad n=0,\ldots,i-1.
\end{eqnarray*}

Since $Y_0>0$ we have $\mathbb{P}[\mathcal{R}_0(i)|\Omega_i]=1$ and therefore, since $\mathcal{R}_n(i)\subseteq\Omega_i$, $\Phi$ takes values on $[0,1]$, and $i\leq N_{\max}$,
\begin{eqnarray*}
\mathbb{P}[\mathcal{R}_i(i)|\Omega_i]\geq\mathbb{P}[\mathcal{R}_i(i)]&=&\mathbb{P}\left[\bigcap_{n=0}^{i}\mathcal{R}_n(i)\right]\\&=&\prod_{n=1}^{i}\mathbb{P}[\mathcal{R}_n(i)|\mathcal{R}_{n-1}(i),\ldots,\mathcal{R}_0(i)]\\
&=&\prod_{n=0}^{i-1}\mathbb{P}[Y_{n+1}>0|\mathcal{R}_n(i)]\\
&\geq&\Phi\left(\frac{1}{\gamma}\left(\frac{Q}{h_{\max}}+\sqrt{h_{\max}}\left(\frac{\alpha}{R\sqrt{\rho}}+\beta R\right)\right)\right)^{i}\\
&\geq&\Phi\left(\frac{1}{\gamma}\left(\frac{Q}{h_{\max}}+\sqrt{h_{\max}}\left(\frac{\alpha}{R\sqrt{\rho}}+\beta R\right)\right)\right)^{N_{\max}},
\end{eqnarray*}
for $i=N_{\min},\ldots,N_{\max}$. Multiplying through by $\mathbb{P}[\Omega_i]$ and applying the Law of Total Probability by summing both sides over $i=N_{\min},\ldots,N_{\max}$ gives
\begin{eqnarray*}
\mathbb{P}[\mathcal{R}_N]&=&\sum_{i=N_{\min}}^{N_{\max}}\mathbb{P}[\mathcal{R}_i(i)|\Omega_i]\mathbb{P}[\Omega_i]\\
&\geq&\sum_{i=N_{\min}}^{N_{\max}}\Phi\left(\frac{1}{\gamma}\left(\frac{Q}{h_{\max}}+\sqrt{h_{\max}}\left(\frac{\alpha}{R\sqrt{\rho}}+\beta R\right)\right)\right)^{N_{\max}}\mathbb{P}[\Omega_i]\\
&=&\Phi\left(\frac{1}{\gamma}\left(\frac{Q}{h_{\max}}+\sqrt{h_{\max}}\left(\frac{\alpha}{R\sqrt{\rho}}+\beta R\right)\right)\right)^{N_{\max}}.
\end{eqnarray*}
Fix $\varepsilon\in(0,1)$, then for all $h_{\max}\in(0,\bar h_{\max}(\varepsilon))$, we have
\begin{equation}\label{eq:increaseT}
\Phi\left(\frac{1}{\gamma}\left(\frac{Q}{h_{\max}}+\sqrt{h_{\max}}\left(\frac{\alpha}{R\sqrt{\rho}}+\beta R\right)\right)\right)^{N_{\max}}\geq 1-\varepsilon.
\end{equation}
To \eqref{eq:increaseT}, apply the following inequality due to \cite{SasChen}
\begin{equation*}
\frac{1}{\sqrt{2\pi}}\int_{-x}^{x}e^{-s^2/2}ds>\sqrt{1-e^{-x^2/2}},\quad x\in\mathbb{R^+},
\end{equation*}
along with the fact that $N_{\max}=\rho T/h_{\max}$, leading us to seek $h_{\max}$ so that
\begin{equation*}
\left(\frac{1}{2}+\frac{1}{2}\sqrt{1-\exp\left(-\frac{\left(\frac{Q}{h_{\max}}+\sqrt{h_{\max}}\left(\frac{\alpha}{R\sqrt{\rho}}+\beta R\right)\right)^2}{2\gamma^2}\right)}\right)^{\frac{\rho T}{h_{\max}}}\geq 1-\varepsilon.
\end{equation*}
Thus we derive the bound
\begin{multline*}
\tilde h_{\max}(\varepsilon)
:=  \sup { \Bigg\{ \bar h\in\left(0,1\right)}:\\
 \frac{Q}{h}+\sqrt{h}\left(\frac{\alpha}{R\sqrt{\rho}}+\beta R\right)\geq\sqrt{\ln\left(1-(2(1-\varepsilon)^{\frac{h}{\rho T}}-1)^2\right)^{-2\gamma^2}},\,h\in(0,\bar h) \Bigg\}.
\end{multline*}
$\bar h_{\max}(\varepsilon)$ is uniquely defined for each $\varepsilon\in(0,1)$ because 
\begin{equation}\label{eq:g}
g(h):=\frac{Q}{h}+\sqrt{h}\left(\frac{\alpha}{R\sqrt{\rho}}+\beta R\right)-\sqrt{-2\gamma^2\ln\left(1-(2(1-\varepsilon)^{h/\rho T}-1)^2\right)}
\end{equation}
is continuous on $\mathbb{R}^+$ with $\lim_{h\to 0^+}g(h)=\infty$, and therefore there is a neighbourhood of zero corresponding to $(0,\bar h_{\max}(\varepsilon))$ within which $g$ is positive.
\end{proof}
Note that if we extend the interval of simulation $[0,T]$ and keep $\varepsilon\in(0,1)$ fixed, there will be a corresponding increase in $N_{\max}$ in \eqref{eq:increaseT}. This will lead to a reduction in the bound $\bar h_{\max}(\varepsilon)$, in a way that is characterised by \eqref{eq:g}. More generally $g(h)$, as defined by \eqref{eq:g} in the proof of Theorem \ref{thm:pos}, provides a practical guide for choosing $h_{\max}$ in order to control the probability of invoking the backstop to avoid negative values.

\begin{example}\label{ex:pos}

Consider two adaptive timestepping strategies based on \eqref{def:PBpos} with $r=1$ and $\rho=2^6,2^8$, each used to simulate a single trajectory of \eqref{eq:equ3} over the interval $[0,1]$ using the adaptive method \eqref{eq:equ15}. 
In each case, we wish to choose $h_{\max}$ so that the probability of requiring the backstop in order to avoid negative values on that trajectory is less than $\varepsilon$, and this will hold for any $Y_0\in(h_{\min},R)$. 

Table \ref{tab:barh} shows the value of $\bar  h(\varepsilon)$ for a range of tolerances $\varepsilon$ for parameter sets where Assumption \ref{assum:fel2} is satisfied ($\kappa\lambda>2\sigma^2$), and where it is not ($\kappa\lambda<2\sigma^2$). The resulting bounds on $h_{\max}$ are determined by substituting all parameters into \eqref{eq:g} and solving $g(h)=0$ for $h$ using the {\tt fsolve} command in Maple with 20 digits of precision. We report the first 4 significant digits in each case, which is sufficient to illustrate the sensitivity of these bounds to the choice of $\rho$ and $\varepsilon$. 

\begin{table}
\def\arraystretch{1.1}
\begin{center}

$\rho=2^6$, $Q=0.015625$, $R=64$\\
\vspace{0.1cm}
\begin{tabular}{|c||c|c|}
\hline
& $\varepsilon$ & $\bar h_{\max}(\varepsilon)$\\
\hline
$\sigma=0.2$ & $10^{-2}$ & $3.594\times 10^{-3}$\\
$\lambda=0.05$ & $10^{-4}$ & $3.547\times 10^{-3}$\\
$\kappa=2$  & $10^{-6}$ & $3.506\times 10^{-3}$\\
\hline
\end{tabular}
\hspace{0.5cm}
\begin{tabular}{|c||c|c|}
\hline
& $\varepsilon$ & $\bar h_{\max}(\varepsilon)$\\
\hline
$\sigma=0.2$ & $10^{-2}$ & $5.454\times 10^{-3}$\\
$\lambda=0.05$ & $10^{-4}$ & $5.341\times 10^{-3}$\\
$\kappa=1$  & $10^{-6}$ & $5.246\times 10^{-3}$\\
\hline
\end{tabular}\\
\vspace{0.1cm}
$\rho=2^8$, $Q=0.00390625$, $R=256$\\
\vspace{0.1cm}
\begin{tabular}{|c||c|c|}
\hline
& $\varepsilon$ & $\bar h_{\max}(\varepsilon)$\\
\hline
$\sigma=0.2$ & $10^{-2}$ & $5.800\times 10^{-4}$\\
$\lambda=0.05$ & $10^{-4}$ & $5.755\times 10^{-4}$\\
$\kappa=2$  & $10^{-6}$ & $5.716\times 10^{-4}$\\
\hline
\end{tabular}
\hspace{0.5cm}
\begin{tabular}{|c||c|c|}
\hline
& $\varepsilon$ & $\bar h_{\max}(\varepsilon)$\\
\hline
$\sigma=0.2$ & $10^{-2}$ & $8.912\times 10^{-4}$\\
$\lambda=0.05$ & $10^{-4}$ & $8.804\times 10^{-4}$\\
$\kappa=1$  & $10^{-6}$ & $8.710\times 10^{-4}$\\
\hline
\end{tabular}
\end{center}
\caption{Bounds on $h_{\max}$ ensuring positivity (without the use of the backstop) of trajectories of \eqref{eq:equ16} with probability at least $1-\varepsilon$ where the path-bounded timestepping strategy satisfies  Definition \ref{def:defYn}. Assumption \ref{assum:fel2} is satisfied ($\kappa\lambda>2\sigma^2$) for tables in the left column, and violated ($\kappa\lambda<2\sigma^2$) for tables in the right column.
}\label{tab:barh}
\end{table}

\end{example}

\section{Numerical simulation}\label{sec:numerics}
Given \eqref{eq:equ1} and its associated transformation \eqref{eq:equ3}, we compare our hybrid adaptive method \eqref{eq:equ15}, referred to in this section as {\tt Explicit Adaptive (EA)}, to a natural semi-implicit variant constructed by replacing the update equation \eqref{eq:equ14} with 
\begin{equation}\label{eq:SIA}
Y_{n+1}=(1-\beta h_{n+1})^{-1}\left[Y_n+h_{n+1}\frac{\alpha}{Y_n}+\gamma\triangle W_{n+1}\right],
\end{equation}
referred to in this section as {\tt Semi-Implicit Adaptive (SIA)}. In both cases we will use the adaptive timestepping strategy given by \eqref{eq:equ12} with $r=1$ (note that we see similar results when $r=2$). We also compare to three fixed step methods: the explicit discretisation of \eqref{eq:equ1} analysed by \cite{9} given by
\[
X_{n+1}=X_n+h\kappa(\lambda-X_n)+\sigma\sqrt{|X_n|}\triangle W_{n+1},
\]
referred to in this section as {\tt Explicit Fixed (EF)}, 
the fully truncated method proposed by \cite{15} given by 
\[
\widetilde{X}_{n+1}=\widetilde{X}_n+h\kappa(\lambda-\widetilde{X}_n^+)+\sigma\sqrt{\widetilde{X}_n^+}\triangle W_{n+1};\quad X_{n+1}=\widetilde{X}_{n+1}^+;\quad \widetilde{X}_0=X_0,
\]
referred to in this section as {\tt Fully Truncated (FT)}, and the drift implicit square root discretisation of \eqref{eq:equ3} proposed and analysed in~\cite{Alphonsi2005,7,Alphonsi2013}, given by
\[
Y_{n+1}=\frac{Y_n+\gamma\triangle W_{n+1}}{2(1-\beta h)}+\sqrt{\frac{(Y_n+\gamma\triangle W_{n+1})^2}{4(1-\beta h)^2}+\frac{\alpha h}{1-\beta h}},
\]
and referred to in this section as {\tt Implicit Fixed (IF)}.

In the first part, we will compare the strong convergence of these methods in the mean stepsize, and the corresponding numerical efficiency. In the second part we explore the dependence on model parameters.  

\subsection{Strong convergence and efficiency}
Throughout the section, we take $\rho=2^6$. We solve using {\tt EA} and {\tt SIA} with values of $h_{\max}=2^{-i}$, $i=4,\ldots,9$ and $M$ sample trajectories to estimate $\sqrt{\expect{|X(T)-X_N|^2}}$, the root mean square error (RMSE), at a final time $T=1$. To compute error estimates we first generate a reference solution using {\tt IF} over a mesh with stepsize $h=2^{-25}$, using a Brownian bridge to ensure values for the adaptive approximations are on the reference trajectory. To ensure that we are comparing adaptive and fixed step schemes of similar average cost, when solving using {\tt IF}, {\tt EF}, and {\tt FT} we take as the fixed step $h_{\text{mean}}$ the
average of all timesteps $h_n^{(m)}$ taken by {\tt EA} over each path and each realisation $\omega_m$, $m=1,\ldots,M$ so that $$h_{\text{mean}}=\frac{1}{M} \sum_{m=1}^M \frac{1}{N^{(\omega_m)}}\sum_{n=1}^{N^{(\omega_m)}} h_{n}^{(\omega_m)}.$$

In Figure \ref{fig:Conv&CPUmean} we examine strong convergence for these methods by plotting RMSE against $h_{\text{mean}}$ with $M=1000$ on a log-log scale, and  efficiency by plotting RMSE against average compute time (cputime) again with $M=1000$. 

In Figure \ref{fig:Conv&CPUmean} (a), Assumption \ref{assum:fel2} holds. The estimated error at each value of $h_{\text{mean}}$ is comparable for all methods except {\tt FT}, and the numerical order appears to be close to one. \CH{For {\tt FT}, the estimated error at each value of $h_{\text{mean}}$ is higher, and the numerical order appears closer to 1/2.} In Figure \ref{fig:Conv&CPUmean} (b) we also see comparable efficiencies as measured by CPU time for this example, \CH{again with the exception of {\tt FT}}. In Figure \ref{fig:Conv&CPUmean} (c), Assumption \ref{assum:fel2} does not hold, and we see first that the numerical order of {\tt EF} has reduced, and the estimated error at each value of $h_{\text{mean}}$ is lowest for {\tt EA} and {\tt SIA}, which also demonstrate the fastest CPU times in Figure \ref{fig:Conv&CPUmean} (d) for lower RMSE values. 

\begin{figure}
\begin{center}
$\begin{array}{@{\hspace{-0.05in}}c@{\hspace{-0.25in}}c}
\scalebox{0.45}{\includegraphics{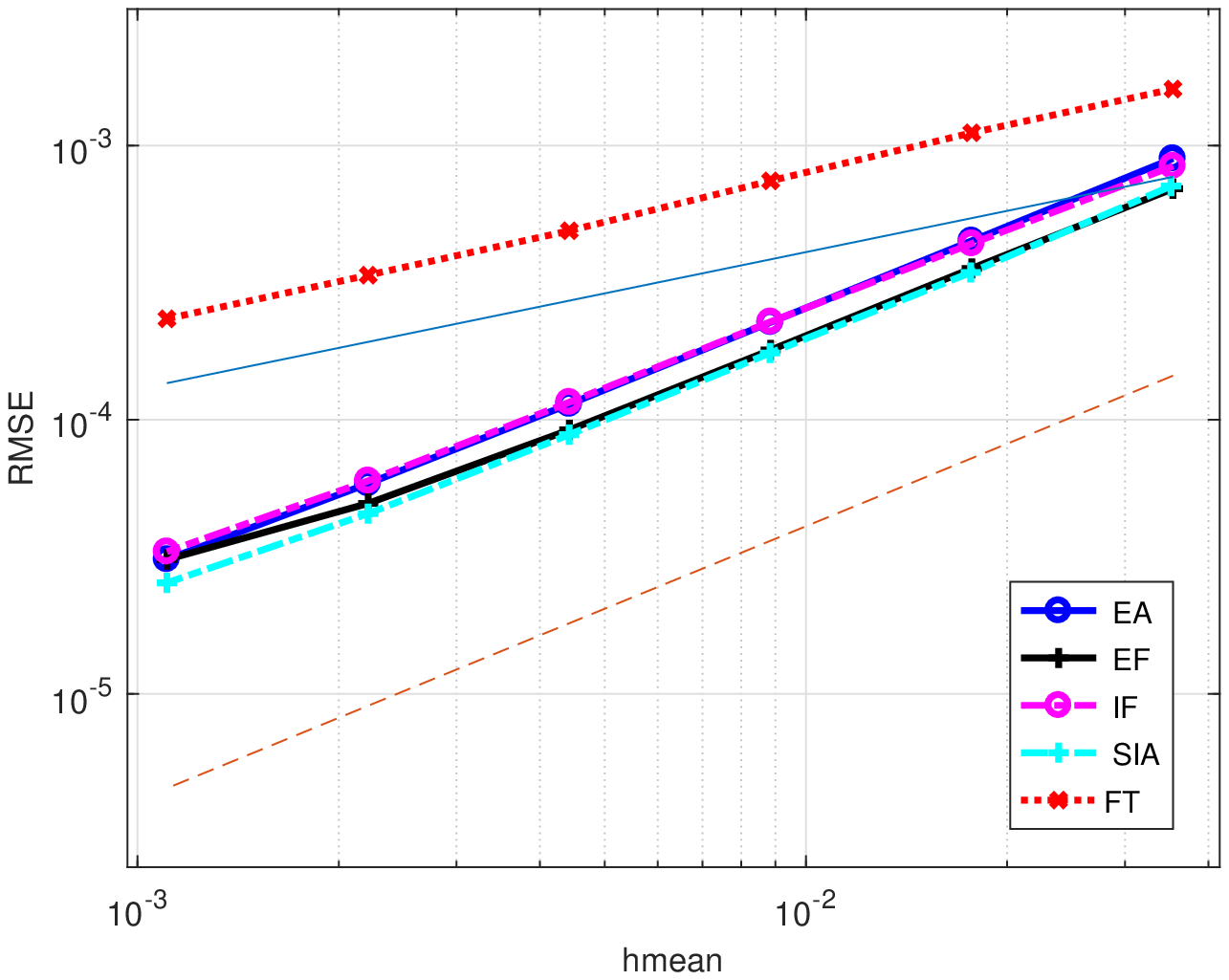}} & \scalebox{0.45}{\includegraphics{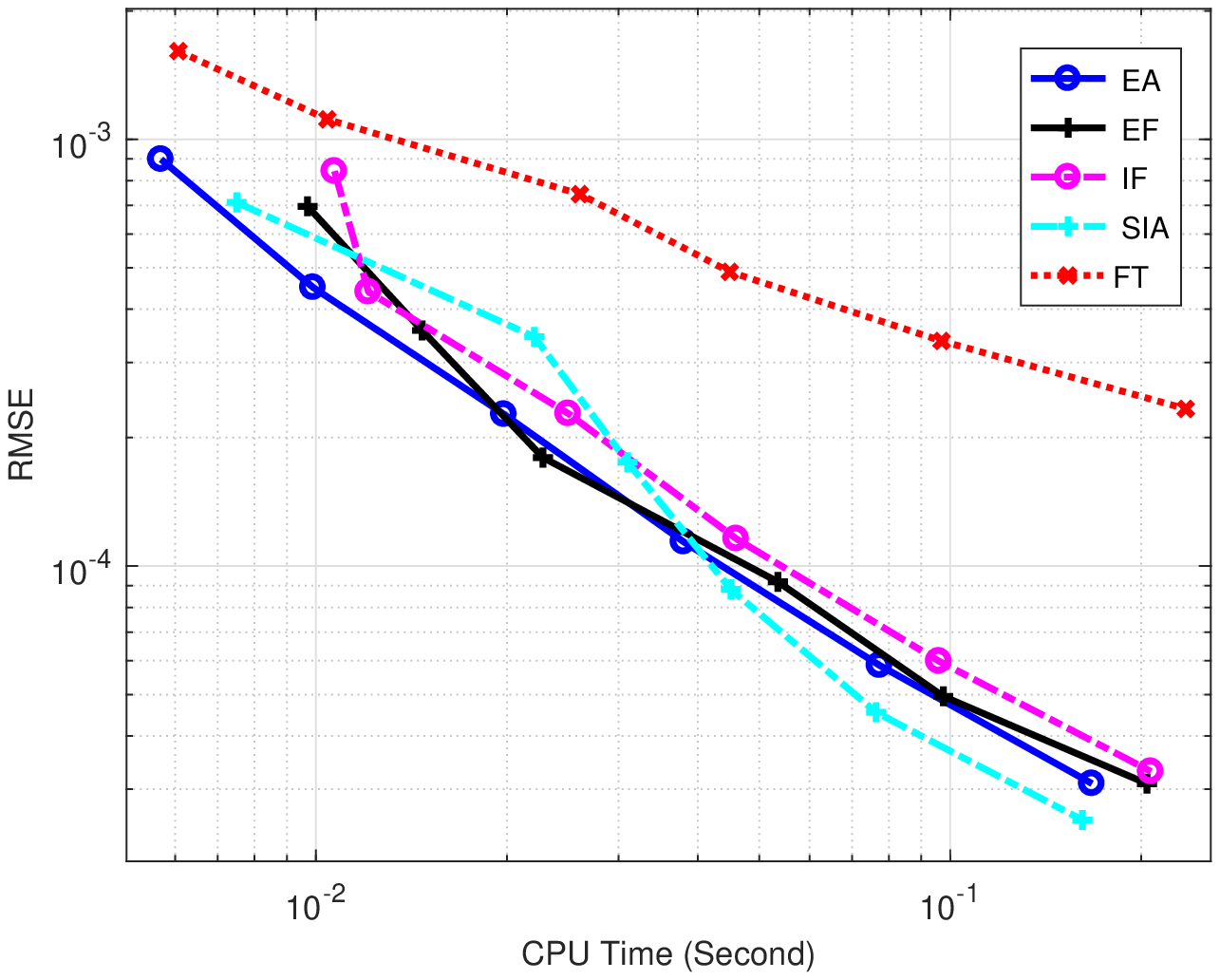}}\\
\mbox{\bf\small (a)} & \mbox{\bf\small (b)}\\
\scalebox{0.45}{\includegraphics{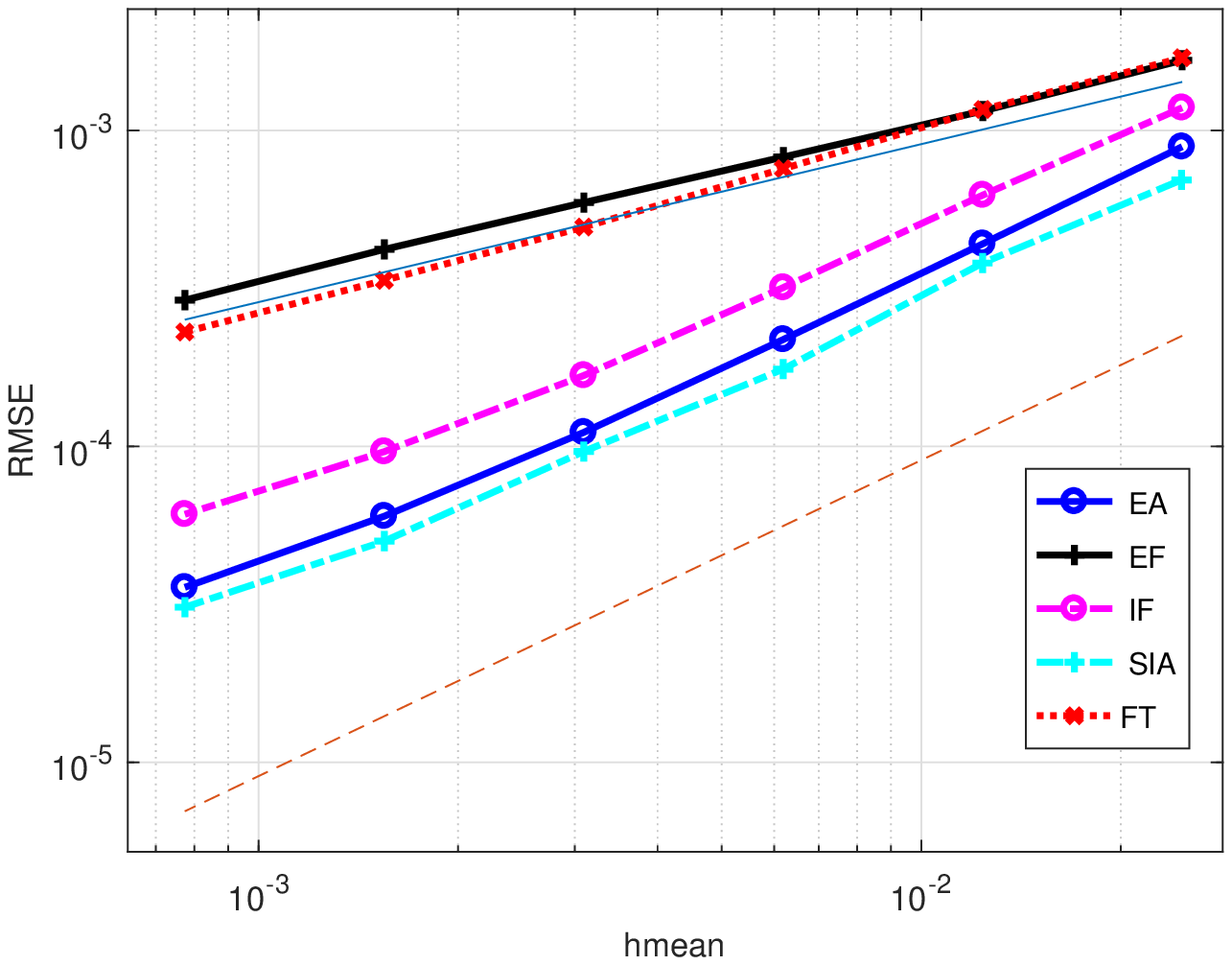}} & \scalebox{0.45}{\includegraphics{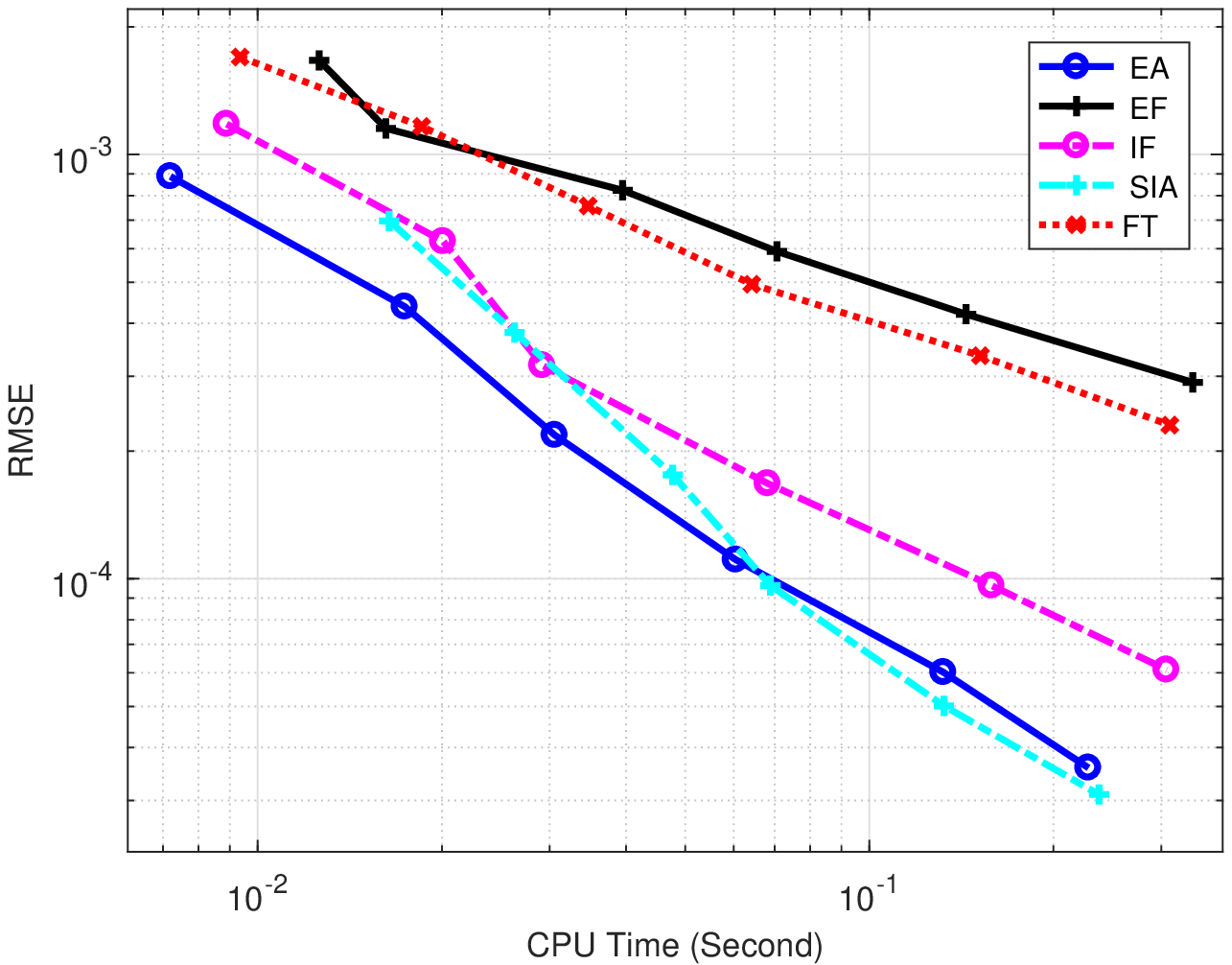}}\\
\mbox{\bf\small (c)} & \mbox{\bf\small (d)}\\
\end{array}$
\end{center}
\caption{Convergence and efficiency of methods applied to \eqref{eq:equ1} with $\lambda=0.05$, $\sigma=0.2$, $Y_0=0.02$. $\kappa=2$ in {\bf (a)} and {\bf (b)}, and $\kappa=1$ in {\bf (c)} and {\bf (d)}. Reference lines of slope $1/2$ and $1$ are provided.}
\label{fig:Conv&CPUmean}
\end{figure}

\subsection{Parameter dependence of the strong convergence rate}
Finally, we investigate numerically the dependence of the rate of strong convergence on the value of the parameter $a:=\sigma^2/(2\kappa\lambda)$. Note that Assumption \ref{assum:fel2} corresponds to $a<0.25$, and Feller's condition corresponds to $a\leq 1$. For $40$ uniformly spaced values of $a$ in the interval $[0.04,1.6]$, we numerically estimate the order of strong convergence in $L_2$ as the slope of the corresponding error over a range of values of $h_{\text{mean}}$ computed by generated strong convergence plots as in Figure \ref{fig:Conv&CPUmean} and using the {\tt polyfit} command in MATLAB to estimate the order of strong convergence for each method. We make the following caveat: for $a\in(1,1.6]$ all numerical schemes presented here are well-defined stochastic processes, even though the SDE \eqref{eq:equ3} is not well defined in that parameter regime. We present the numerically estimated error in terms of $X$ outside the Feller regime using the reference solution generated by {\tt IF}, though it is not known if {\tt IF} converges in that regime with nonzero rate.

In Figure \ref{fig:delta} we observe that {\tt EA} maintains the highest numerical order of convergence: at or close to one while Feller's condition holds. The reduction in order outside of this region, which is visible for all methods, occurs more sharply in Figure \ref{fig:delta} (b) when $\kappa$ is small. This is followed by {\tt SIA}, which maintains a numerical order of convergence close to {\tt EA} when Assumption \ref{assum:fel2} holds, but reduces more quickly outside this region. The difference is more pronounced in Figure \ref{fig:delta} (a), and this may be because updates using the {\tt SIA} method are subject to a damping factor $(1-\beta h_{n+1})^{-1}$, as can be seen in \eqref{eq:SIA}, which has greater effect for larger values of $\kappa$.  Finally, we observe for {\tt EF} and {\tt FT} an uptick in convergence rate as $a$ approaches zero, and this is consistent with the notion that it should display order one convergence in the absence of noise. 

\begin{figure}
\begin{center}
$\begin{array}{@{\hspace{-0.1in}}c@{\hspace{-0.2in}}c}
\scalebox{0.45}{\includegraphics{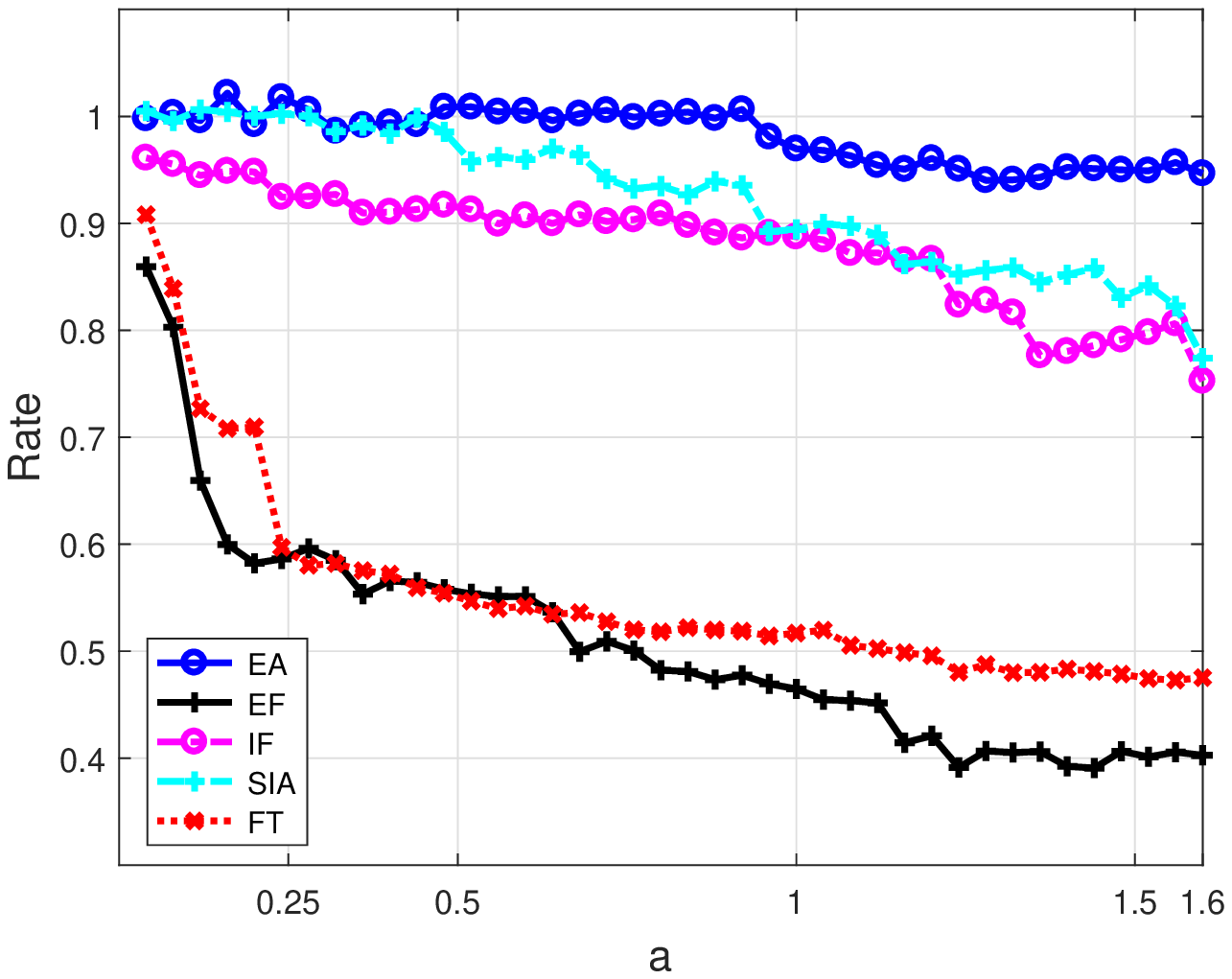}} & \scalebox{0.45}{\includegraphics{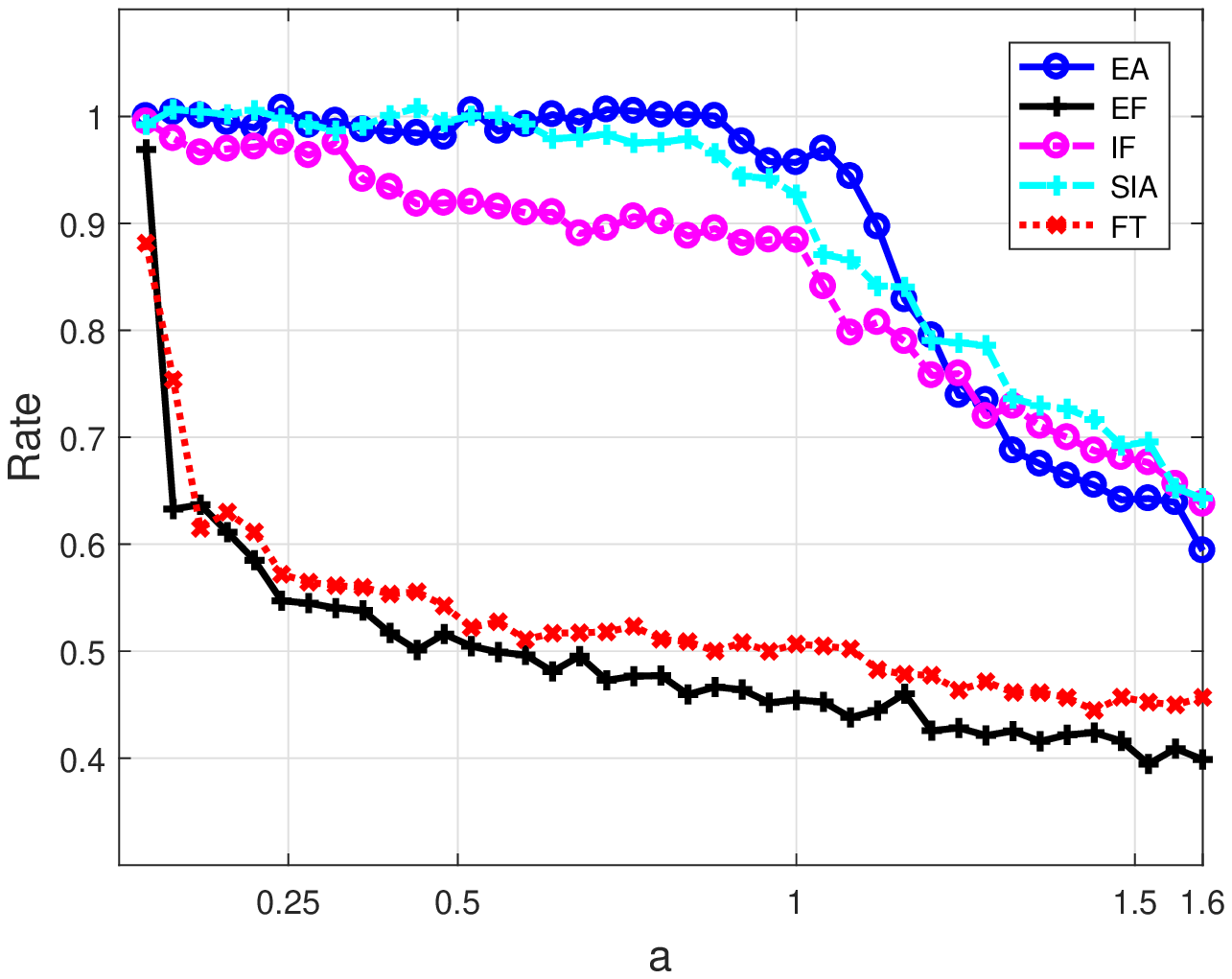}}\\
\mbox{\bf\small (a)} & \mbox{\bf\small (b)}
\end{array}$
\caption{Estimation of strong convergence order for methods applied to \eqref{eq:equ1} with $\lambda=0.05$, $Y_0=0.02$, $\kappa=2$ in {\bf (a)}, and $\kappa=0.2$ in {\bf (b)}. Here, $a=\sigma^2/(2\kappa\lambda)$. In both cases, Assumption \ref{assum:fel2} holds to the left of the vertical line at $a=0.25$, and Feller's condition holds to the left of the vertical line at $a=1$.}\label{fig:delta}
\end{center}
\end{figure}

Figure \ref{fig:BFTP} demonstrates how frequently the backstop was invoked in the production of Figure \ref{fig:delta}, where we separately track usage to ensure positivity and usage to bound below the stepsize at $h_{\min}$. We see that when Assumption \ref{assum:fel2} holds, we do not require the backstop to avoid negative values, though as we move to the boundary of that region we do start to use it to bound the stepsize at $h_{\min}$ for a small proportion of steps. Note also that usage to avoid negative values increases with $a$ when $\kappa=2$ only, whereas usage to bound the timestep increases with $a$ for both $\kappa=2,0.2$, more rapidly in the latter case. 

\begin{figure}
\begin{center}
$\begin{array}{@{\hspace{-0.1in}}c@{\hspace{-0.2in}}c}
\scalebox{0.45}{\includegraphics{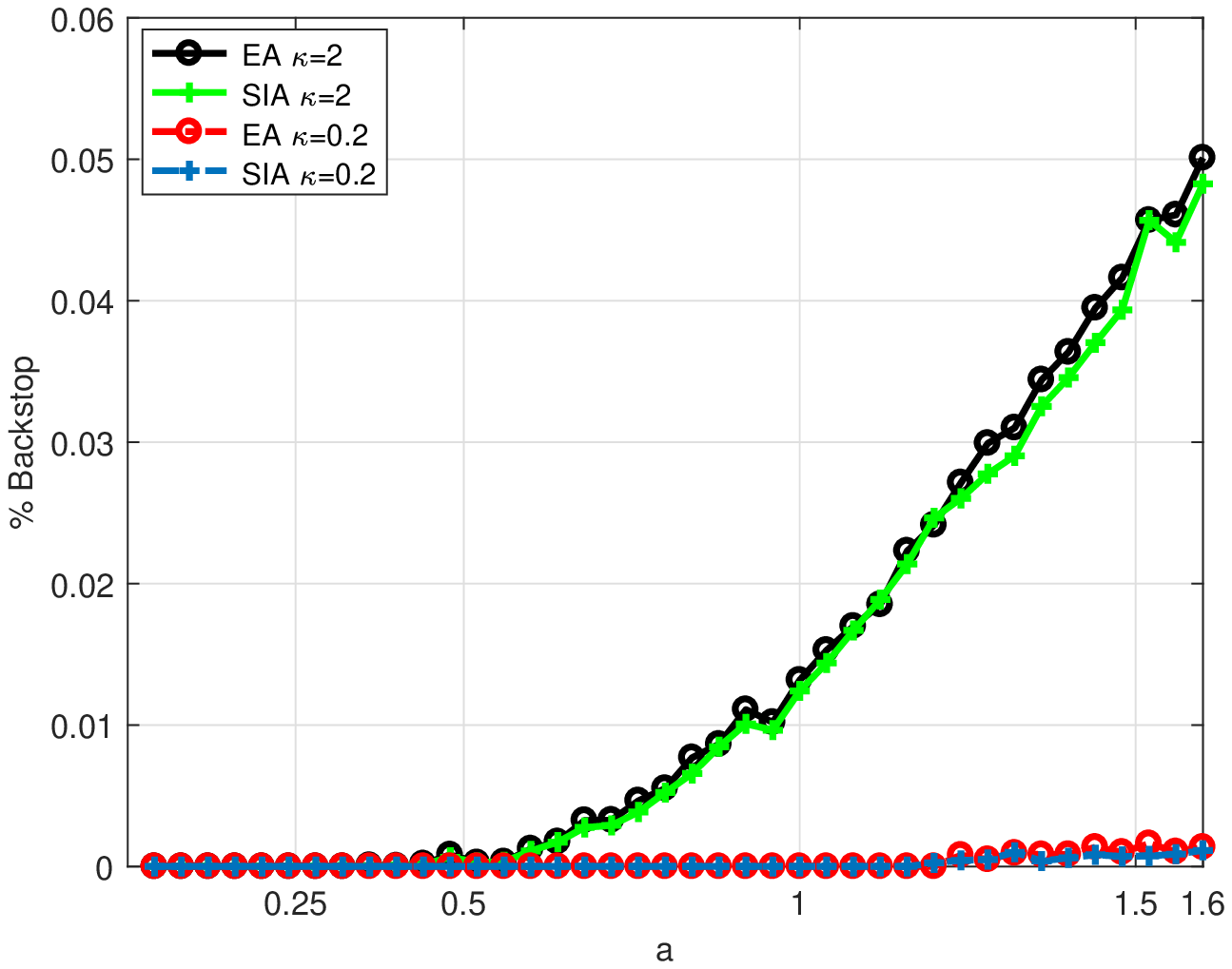}} & \scalebox{0.45}{\includegraphics{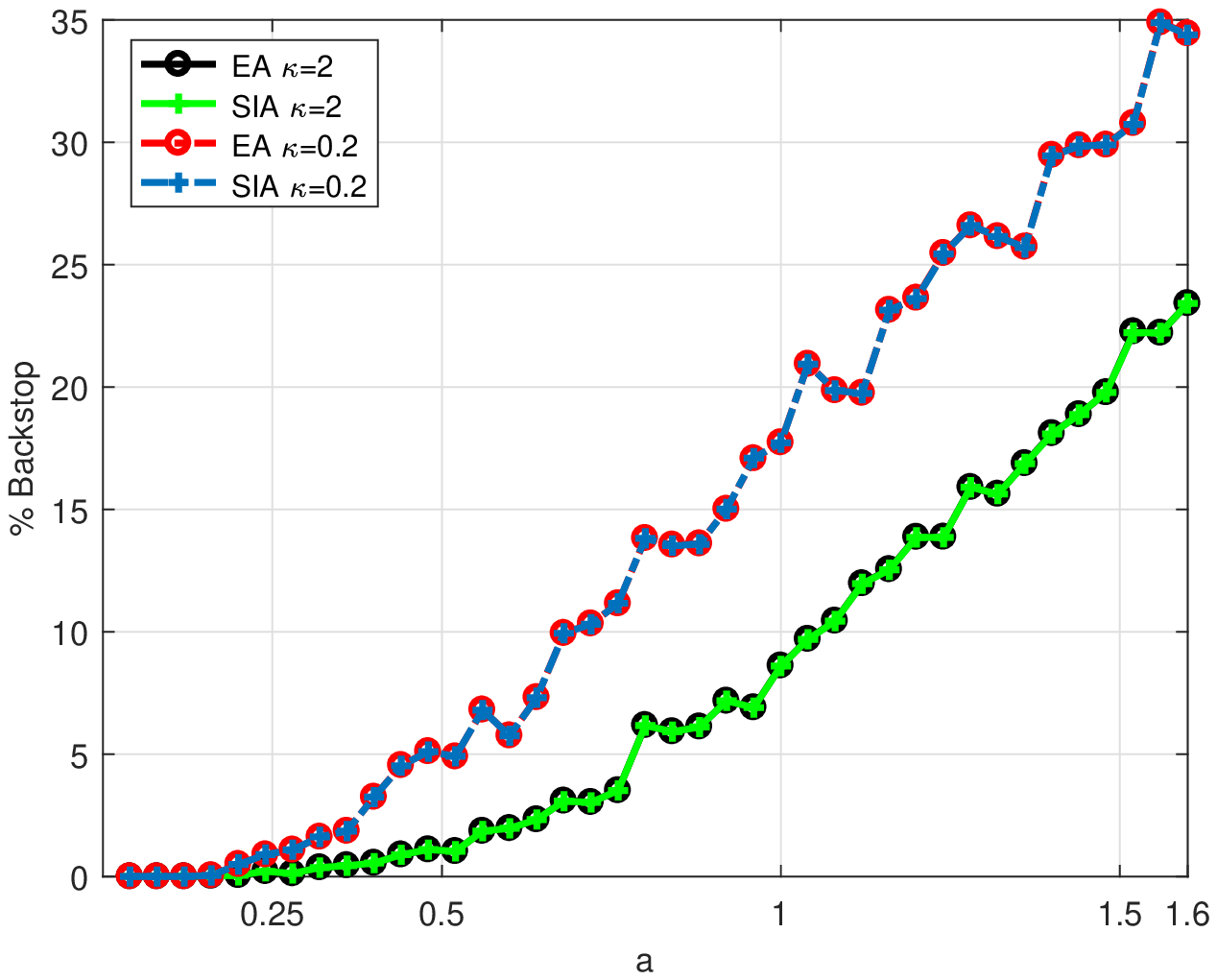}}\\
\mbox{\bf\small (a)}: \text{Avoid }Y_n<0 & \mbox{\bf\small (b)}: \text{Ensure }h_n=h_{\min}
\end{array}$
\caption{Percentage of times the backstop was invoked in the production of Figure \ref{fig:delta} in order to {\bf (a)} avoid a negative value or {\bf (b)} bound the timestep from below by $h_{\min}$.}\label{fig:BFTP}
\end{center}
\end{figure}

\section*{Acknowledgements}
The authors are grateful to Professor Alexandra Rodkina, of the University of the West Indies at Mona, Jamaica, and Ms Fandi Sun, of Heriot-Watt University, Edinburgh, UK, for useful discussion in preparing this work. The manuscript comprises the central part of the PhD project of the final author, Heru Maulana, who tragically passed away in December 2020 shortly after the submission of the manuscript, and we dedicate it to his memory.

\section*{Declaration of Interest}
The first and third authors were supported by a grant from \textit{Lembaga Pengelola Dana Pendidikan (LPDP) Republik Indonesia}.

\end{document}